\documentclass{article}
\usepackage{fullpage}
\usepackage{amsthm,amssymb,amsmath}  
\usepackage{authblk}
\usepackage{tikz} 

\usepackage{mathtools, eqparbox}%

\usepackage{wrapfig}

\usepackage{todonotes}
\usepackage{subcaption}
\usepackage{nicefrac}
\usepackage{amsthm,amssymb,amsmath}  
\usepackage{xspace,enumerate}
\usepackage[utf8]{inputenc}
\usepackage{thmtools}
\usepackage{thm-restate}
\usepackage{todonotes}
\usepackage{hyperref}
\usepackage{verbatim}
\usepackage{multirow}
\usepackage{url}
\usepackage{mathtools}
\usepackage{amsthm}
\usepackage[utf8]{inputenc}
\usepackage{color}
\usepackage{caption}

\usepackage{subcaption}
\usepackage{todonotes}
\usepackage{algorithm}
\usepackage[noend]{algpseudocode}
\usepackage[normalem]{ulem}
\usetikzlibrary{calc,snakes,shapes,arrows.meta}
\usepackage{booktabs}

\newcommand{\cO}{\mathcal{O}}
\newcommand{\ctO}{\mathcal{\tilde{O}}}

\newcommand{\oneAP}{1\text{-}AP}

\newcommand{\defproblem}[3]{
\vspace{2mm}
\noindent\fbox{
   \begin{minipage}{0.96\textwidth}
   \textsc{#1}\\
   {\bf{Input:}} #2  \\
   {\bf{Output:}} #3
   \end{minipage}
   }
   \vspace{2mm}
}
\def\dd{\mathinner{.\,.}}

  \theoremstyle{plain}
  \newtheorem{theorem}{Theorem}
  \newtheorem{lemma}{Lemma}  
  \newtheorem{corollary}[theorem]{Corollary} 
  \newtheorem{proposition}[theorem]{Proposition}

  \theoremstyle{definition}
  \newtheorem{definition}{Definition}

  \newtheorem{remark}[definition]{Remark}
  {\bfseries}{\itshape}

\newif\ifcomment\commentfalse
\def\commentON{\commenttrue}

\long\outer\def\bcl#1\ecl{{\ifcomment \sloppy  \textcolor{red}{{{#1}}}\fi }}
\long\outer\def\bca#1\eca{{\ifcomment \sloppy  \textcolor{purple}{{{#1}}}\fi }}

\long\outer\def\BCL#1\ECL{{\ifcomment \sloppy \par \textcolor{blue}{\#  \dotfill
{\textsc{#1}} \dotfill \#} \par \fi }}

\long\outer\def\BCA#1\ECA{{\ifcomment \sloppy \par \textcolor{green}{\#  \dotfill
{\textsc{#1}} \dotfill \#} \par \fi }}

\commentON

\title{Elastic-Degenerate String Matching with 1 Error\thanks{The work in this paper is supported in part by: the Netherlands Organisation for Scientific Research (NWO) through project OCENW.GROOT.2019.015 ``Optimization for and with Machine Learning (OPTIMAL)'' and Gravitation-grant NETWORKS-024.002.003; the PANGAIA and ALPACA projects that have received funding from the European Union’s Horizon 2020 research and innovation programme under the Marie Skłodowska-Curie grant agreements No 872539 and 956229, respectively; and the MUR - FSE REACT EU - PON R\&I 2014-2020.}}

\author[1]{Giulia Bernardini}
\author[2]{Esteban Gabory}
\author[2,3,4]{Solon P.\ Pissis}
\author[2,3,4]{Leen Stougie}
\author[2]{Michelle Sweering}
\author[2]{Wiktor Zuba}

\affil[1]{University of Trieste, Trieste, Italy}
\affil[2]{CWI, Amsterdam, The Netherlands}
\affil[3]{Vrije Universiteit, Amsterdam, The Netherlands}
\affil[4]{INRIA-Erable, France}

\date{\today}

\begin{document}

\maketitle

\begin{abstract}
An elastic-degenerate (ED) string is a sequence of $n$ finite sets of strings of total length $N$, introduced to represent a set of related DNA sequences, also known as a \emph{pangenome}. The ED string matching (EDSM) problem consists in reporting all occurrences of a pattern of length $m$ in an ED text. The EDSM problem has recently received some attention by the combinatorial pattern matching community, culminating in an $\ctO(nm^{\omega-1})+\cO(N)$-time algorithm [Bernardini et al., SIAM J. Comput. 2022], where $\omega$ denotes the matrix multiplication exponent and the $\ctO(\cdot)$ notation suppresses polylog factors.
In the $k$-EDSM problem, the approximate version of EDSM, we are asked to report all pattern occurrences with at most $k$ errors. $k$-EDSM can be solved in $\cO(k^2mG+kN)$ time, under edit distance, or $\cO(kmG+kN)$ time, under Hamming distance, where $G$ denotes the total number of strings in the ED text [Bernardini et al., Theor. Comput. Sci. 2020]. Unfortunately, $G$ is only bounded by $N$, and so even for $k=1$, the existing algorithms run in $\Omega(mN)$ time in the worst case. In this paper we make progress in this direction. We show that $1$-EDSM can be solved in $\cO((nm^2 + N)\log m)$ or $\cO(nm^3 + N)$ time under edit distance. For the decision version of the problem, we present a faster $\cO(nm^2\sqrt{\log m} + N\log\log m)$-time algorithm. We also show that $1$-EDSM can be solved in $\cO(nm^2 + N\log m)$ time under Hamming distance. Our algorithms for edit distance rely on non-trivial reductions from $1$-EDSM to special instances of classic computational geometry problems (2d rectangle stabbing or 2d range emptiness), which we show how to solve efficiently. In order to obtain an even faster algorithm for Hamming distance, we rely on employing and adapting the $k$-errata trees for indexing with errors [Cole et al., STOC 2004].
\end{abstract}
\section{Introduction}
String matching (or pattern matching) is a fundamental task in computer science, for which several linear-time algorithms are known~\cite{DBLP:books/daglib/0020103}. 
It consists in finding all occurrences of a short string, known as the \emph{pattern},
in a longer string, known as the \emph{text}. Many representations have been introduced over the years to account for unknown or uncertain letters in the pattern or in the text, a phenomenon that often occurs in real data. In the context of computational biology, for example, the IUPAC notation~\cite{IUPAC} is used to represent locations of a DNA sequence for which several alternative nucleotides are possible. Such a notation can encode the consensus of a population of DNA sequences~\cite{PanGenomeConsortium18,wabi18,DBLP:journals/fuin/AlzamelABGIPPR20,IndetCPM20} in a gapless multiple sequence alignment (MSA).

Iliopoulos et al.~generalized these representations in~\cite{DBLP:journals/iandc/IliopoulosKP21} to also encode insertions and deletions (gaps) occurring in MSAs by introducing the notion of elastic-degenerate strings. An \emph{elastic-degenerate} (ED) string $\tilde{T}$ over an alphabet $\Sigma$ is a sequence of finite subsets of $\Sigma^*$ (which includes the empty string $\varepsilon$), called \emph{segments}. The total number of segments is the \emph{length} of the ED string, denoted by $n=\lvert\tilde{T}\rvert$; and the total number of letters (including symbol $\varepsilon$) in all segments is the \emph{size} of the ED string, denoted by $N=\lVert\tilde{T}\rVert$. Inspect Figure~\ref{fig:EDtext} for an example.

\begin{figure}[t]
    \centering
    \includegraphics[width=1.0\textwidth]{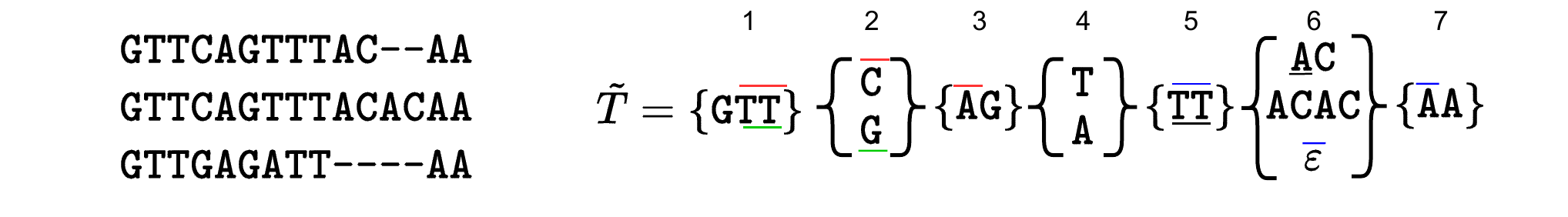}
    \caption{An MSA of three sequences and its (non-unique) representation $\tilde{T}$ as an ED string of length $n=7$ and size $N=20$. The only two \emph{exact} occurrences of $P=\texttt{TTA}$ in $\tilde{T}$ end at positions 6 (black underline) and 7 (blue overline); a \emph{1-mismatch} occurrence of $P$ in $\tilde{T}$ ends at position 2 (green underline); and a \emph{1-error} occurrence of $P$ in $\tilde{T}$ ends at position 3 (red overline). Note that other $1$-error and $1$-mismatch occurrences of $P$ in $\tilde{T}$ exist (e.g., ending at positions 1 and 5).} 
    \label{fig:EDtext}
\end{figure}

A natural problem is to find all occurrences of a standard (non-degenerate) pattern $P$ in an ED text $\tilde{T}$, called the ED string matching (EDSM) problem in the literature. After the simple polynomial-time algorithm proposed by Iliopoulos et al.~\cite{DBLP:journals/iandc/IliopoulosKP21}, a series of results have been published for EDSM. The results for EDSM summarized in Table~\ref{tab:EDSM} have a \emph{linear dependency} on the size $N$ of the ED text, a highly desirable property. (A different line of research exists, which waives the linear-dependency restriction, and employs bit-vector techniques to speed up the computation specifically for short patterns~\cite{DBLP:conf/cpm/GrossiILPPRRVV17,DBLP:conf/wea/PissisR18,sopang}.) 
In Table~\ref{tab:EDSM}, $m$ is the length of the pattern, $n$ is the length of the ED text, $N$ is its size, and $\omega$ is the matrix multiplication exponent. These algorithms are also \emph{on-line}: the ED text is read segment-by-segment and occurrences are reported as soon as the last segment they overlap is processed. Grossi et al.~\cite{DBLP:conf/cpm/GrossiILPPRRVV17} presented an $\cO(nm^2+N)$-time algorithm for EDSM. This was later improved by Aoyama et al.~\cite{DBLP:conf/cpm/AoyamaNIIBT18}, who employed fast Fourier transform to improve the time complexity of EDSM to $\cO(nm^{1.5}\sqrt{\log m}+N)$. Bernardini et al.~\cite{bernardini_et_al:LIPIcs:2019:10597} then presented a lower bound conditioned on Boolean Matrix Multiplication suggesting that it is unlikely to solve EDSM by a combinatorial algorithm in $\cO(nm^{1.5-\epsilon}+N)$ time, for any $\epsilon>0$. This was an indication that fast matrix multiplication may improve the time complexity of EDSM. Indeed, Bernardini et al.~\cite{bernardini_et_al:LIPIcs:2019:10597}
presented an $\cO(nm^{1.381} + N)$-time algorithm, which they subsequently improved to an $\tilde{\cO}(nm^{\omega-1})+ \cO(N)$-time algorithm~\cite{elasticSICOMP}, both using fast matrix multiplication, thus breaking through the conditional lower bound for EDSM.

\begin{table}[t]
\[
\small
\begin{array}{lll}
\toprule
\textbf{EDSM}&  \textbf{Features} & \textbf{Running time}\\
\midrule
 \textnormal{Grossi et al.~\cite{DBLP:conf/cpm/GrossiILPPRRVV17}} & \textnormal{Combinatorial} & \cO(nm^2+N)\\
 \textnormal{Aoyama et al.~\cite{DBLP:conf/cpm/AoyamaNIIBT18}} & \textnormal{Fast Fourier transform} & \cO(nm^{1.5}\sqrt{\log m}+N)\\
 \textnormal{Bernardini et al.~\cite{bernardini_et_al:LIPIcs:2019:10597}} & \textnormal{Fast matrix multiplication} & \cO(nm^{1.381} + N)\\ 
  \textnormal{Bernardini et al.~\cite{elasticSICOMP}} & \textnormal{Fast matrix multiplication} & \tilde{\cO}(nm^{\omega-1})+\cO(N)\\ 
\bottomrule
\end{array}
\]
\caption{The upper-bound landscape of the EDSM problem. The term ``combinatorial'' is arguably not well-defined; lower bounds conditioned on Boolean Matrix Multiplication often indicate that other techniques, including fast matrix multiplication, may be employed to obtain improved bounds for a specific problem. This is the case for EDSM.}\label{tab:EDSM}
\end{table}

\paragraph{Our Results and Techniques}

\begin{table}[t]
\[
\small
\begin{array}{lll}
\toprule
\textbf{Approximate EDSM}& \textbf{Features} & \textbf{Running time} \\
\midrule
 \textnormal{Bernardini et al.~\cite{tcs-ed2020}} & \textnormal{$k$ errors} & \cO(k^2mG+kN)\\
     \textnormal{\bf This work} & \textnormal{1 error} & \cO(nm^3 + N)\\ 
  \textnormal{\bf This work} & \textnormal{1 error} & \cO((nm^2 + N)\log m)\\ 
    \textnormal{\bf This work} & \textnormal{1 error (decision)} & \cO(nm^2\sqrt{\log m} + N\log\log m)\\ 
    \midrule
     \textnormal{Bernardini et al.~\cite{tcs-ed2020}} & \textnormal{$k$ mismatches} & \cO(kmG+kN)\\
     \textnormal{\bf This work} & \textnormal{1 mismatch} & \cO(nm^3 + N)\\
         \textnormal{\bf This work} & \textnormal{1 mismatch} & \cO(nm^2 + N\log m)\\
\bottomrule
\end{array}
\]
\caption{The state of the art results for approximate EDSM and our new results for $k=1$. Note that $n\leq G \leq N$. All algorithms underlying these results are combinatorial and all the reporting algorithms  are on-line.}\label{tab:results}
\end{table}

In string matching, a single extra or missing letter in the pattern or in a potential occurrence results in missing (many or all) occurrences. Hence, many works are focused on approximate string matching for standard strings~\cite{DBLP:journals/tcs/LandauV86,DBLP:journals/jcss/LandauV88,DBLP:journals/siamcomp/ColeH02,DBLP:journals/jal/AmirLP04,DBLP:conf/icalp/GawrychowskiU18,DBLP:conf/focs/Charalampopoulos20}. For approximate $k$-EDSM, Bernardini et al.~\cite{tcs-ed2020} presented an on-line $\cO(k^2mG+kN)$-time algorithm under edit distance and an on-line $\cO(kmG+kN)$-time algorithm under Hamming distance, where $k$ is the maximum allowed number of errors (edits) or mismatches, respectively, and $G$ is the total number of strings in all segments. Unfortunately, $G$ is only bounded by $N$, and so even for $k=1$, the existing algorithms run in $\Omega(mN)$ time in the worst case. 

Let us remark that the special case of $k=1$ is not interesting for approximate string matching on standard strings: the existing algorithms have a polynomial dependency on $k$ and a linear dependency on the length $n$ of the text, and thus for $k=1$ we trivially obtain $\cO(n)$-time algorithms under edit or Hamming distance. However, this is not the case for other string problems,
such as text indexing with errors, where the first step was to design a data structure for $1$ error~\cite{DBLP:journals/jal/AmirKLLLR00}. The next step, extending it to $k$ errors, required the development of new highly non-trivial techniques and incurred some exponential factor with respect to $k$~\cite{cole_dictionary_2004}.
Interestingly, $k$-EDSM seems to be the same case, which highlights the main theoretical motivation of this paper. In Table~\ref{tab:results}, we summarize the state of the art for approximate EDSM and our new results for $k=1$. Note that the reporting algorithms underlying our results are also \emph{on-line}.

Indeed, to arrive at our main results, we design a rich non-trivial combination of algorithmic techniques. Our algorithms for edit distance rely on non-trivial reductions from $1$-EDSM to special instances of classic computational geometry problems (2d rectangle stabbing or 2d range emptiness), which we show how to solve efficiently. In order to obtain an even faster algorithm for Hamming distance, we also rely on employing and adapting the $k$-errata trees of Cole et al.~for text indexing with $k$ errors~\cite{cole_dictionary_2004}.

The combinatorial algorithms we develop here for approximate EDSM are good in the following sense. First, the running times of our algorithms do not depend on $G$, a highly desirable property. 
Specifically, all of our results replace $m\cdot G$ by an $n\cdot \text{poly}(m)$ factor.
Second, our $\ctO(nm^2 + N)$-time algorithms are at most one $\log m$ factor slower than $\cO(nm^2 + N)$, the best-known bound obtained by a combinatorial algorithm (not employing fast Fourier transforms) for \emph{exact} EDSM~\cite{DBLP:conf/cpm/GrossiILPPRRVV17}. Notably, for Hamming distance, we show an $\cO(nm^2 + N\log m)$-time algorithm. Last, our $\cO(nm^3 + N)$-time algorithms have a linear dependency on $N$, another highly desirable property (at the expense of an extra $m$-factor).

\paragraph{Other Related Work}

The main motivation to consider ED strings is that they can be used to represent a \emph{pangenome}: a collection of closely-related genomic sequences that are meant to be analyzed together~\cite{PanGenomeConsortium18}. Several other pangenome representations have been proposed in the literature, mostly graph-based ones; see~\cite{DBLP:conf/gbrpr/CarlettiFG0RV19} for a comprehensive overview by Carletti et al. Compared to these graph-based representations, ED strings have at least two main advantages in the context of string matching, as they support: (i) simple on-line string matching; and (ii) (deterministic) subquadratic in $m$ string matching~\cite{DBLP:conf/cpm/AoyamaNIIBT18,bernardini_et_al:LIPIcs:2019:10597,elasticSICOMP}.

Similar in spirit to ED strings, and to the restricted notion of \emph{generalized degenerate} strings, in which strings of different lengths cannot be in the same segment~\cite{wabi18,DBLP:journals/fuin/AlzamelABGIPPR20}, is the representation of pangenomes via \emph{founder graphs}. The idea behind founder graphs is that a multiple alignment of few \emph{founder sequences} can be used to approximate the input MSA, with the feature that each row of the MSA is a recombination of the founders. Unlike ED strings, that are believed not to be efficiently indexable~\cite{DBLP:conf/spire/Gibney20} (and indeed their value is to enable fast on-line string matching algorithms), some subclasses of founder graphs are, and a recent line of research is devoted to constructing and indexing such structures~\cite{makinen2020linear,DBLP:conf/isaac/EquiNACTM21}. Like founder graphs, ED strings support the recombination of different rows of the MSA between consecutive columns.

\paragraph{Paper Organization}

In Section~\ref{sec:prel}, we provide the necessary definitions and notation, we describe the basic layout of the developed algorithms, and we formally state our main results. In Section~\ref{sec:edit}, we present our solutions under edit distance; and in Section~\ref{sec:ham}, we present our improvement for the special case of Hamming distance. In Section~\ref{sec:con}, we conclude this work with some basic open questions for future work.

\section{Preliminaries}\label{sec:prel}

We start with some basic definitions and notation following~\cite{DBLP:books/daglib/0020103}. Let $X=X[1]\ldots X[n]$ be a \emph{string} of length $|X|=n$ over an ordered alphabet $\Sigma$ whose elements are called \textit{letters}. 
The \emph{empty string} is the string of length $0$; we denote it by $\varepsilon$.
For any two positions $i$ and $j\geq i$ of $X$, $X[i\dd j]$ is the \emph{fragment} of $X$ starting at position $i$ and ending at position $j$. The fragment $X[i\dd j]$ is an \emph{occurrence} of the underlying \emph{substring} $P=X[i]\ldots X[j]$; we say that $P$ occurs at \emph{position} $i$ in $X$. A \emph{prefix} of $X$ is a fragment of the form $X[1\dd j]$ and a \emph{suffix} of $X$ is a fragment  of the form $X[i\dd n]$. By $XY$ or $X\cdot Y$ we denote the \emph{concatenation} of two strings $X$ and $Y$, i.e., $XY=X[1]\ldots X[|X|]Y[1]\ldots Y[|Y|]$. Given a string $X$ we write $X^R=X[|X|]\ldots X[1]$ for the \emph{reverse} of $X$. Given two strings $X$ and $Y$ we write $\textsf{LCP}(X,Y)$ for the length of their \emph{longest common prefix}, namely for the integer $\max(\{i,\ X[1\dd i]=Y[1 \dd i]\})$, or $0$ if $X[1]\neq Y[1]$.

An \emph{elastic-degenerate string} (ED string) $\tilde{T}=\tilde{T}[1]\ldots \tilde{T}[n]$ over an alphabet $\Sigma$ is a sequence of $n=|\tilde{T}|$ finite sets, called \emph{segments}, such that for every position $i$ of $\tilde{T}$  we have that $\tilde{T}[i]\subset\Sigma^*$. 
By $N=||\tilde{T}||$ we denote the total length of all strings in all segments of $\tilde{T}$, which we call the \emph{size} of $\tilde{T}$; more formally, $N=\sum^{n}_{i=1}\sum_{j=1}^{|\tilde{T}[i]|}|\tilde{T}[i][j]|$, where by $\tilde{T}[i][j]$ we denote the $j$th string of $\tilde{T}[i]$. (As an exception, we also add $1$ to account for empty strings: if $\tilde{T}[i][j]=\varepsilon$, then we have that $|\tilde{T}[i][j]|=1$.)
Given two sets of strings $S_1$ and $S_2$, their \emph{concatenation} is $S_1\cdot S_2=\{XY\ |\ X\in S_1, Y\in S_2\}$. For an ED string $\tilde{T}=\tilde{T}[1]\ldots \tilde{T}[n]$, we define the \emph{language} of $\tilde{T}$ as $\mathcal{L}(\tilde{T})=\tilde{T}[1]\cdot\ldots\cdot \tilde{T}[n]$. Given a set $S$ of strings we write $S^R$ for the set $\{X^R\mid X\in S\}$. For an ED string $\tilde{T}=\tilde{T}[1]\ldots \tilde{T}[n]$ we write $\tilde{T}^R$ for the ED string $\tilde{T}[n]^R\ldots \tilde{T}[1]^R$.

Given a string $P$ and an ED string $\tilde{T}$, we say that $P$ \emph{matches} the fragment $\tilde{T}[j\dd j']=\tilde{T}[j]\ldots \tilde{T}[j']$ of $\tilde{T}$, or that an \emph{occurrence} of $P$ \emph{starts} at position $j$ and \emph{ends} at position $j'$ in $\tilde{T}$ if there exist two strings $U,V$, each of them possibly empty, 
such that $P=P_j \cdot \ldots \cdot P_{j'}$, where $P_{i}\in \tilde{T}[i]$, for every $j< i < j'$, $U\cdot P_j\in \tilde{T}[j]$, and $P_{j'}\cdot V\in \tilde{T}[j']$ (or $U\cdot P_j\cdot V\in \tilde{T}[j]$ when $j=j'$).
Strings $U,V$ and $P_{i}$, for every $j\le i \le j'$, specify an \emph{alignment} of $P$ with $\tilde{T}[j\dd j']$. For each occurrence of $P$ in $\tilde{T}$, the alignment is, in general, not unique. 
In Figure~\ref{fig:EDtext}, $P=\texttt{TTA}$ matches $\tilde{T}[5\dd 6]$ with two alignments: both have $U=\varepsilon$, $P_5=\texttt{TT}$, $P_6=\texttt{A}$, and $V$ is either \texttt{C} or \texttt{CAC}.


We want to accept matches with edit distance at most $1$ according to the following standard definition:

\begin{definition}\label{def:ed}
Given two strings $P$ and $Q$ over an alphabet $\Sigma$, we define the \emph{edit distance} $d_E(P,Q)$ between $P$ and $Q$ as the length $\ell$ of a shortest sequence of string operations $\pi_1,\ldots,\pi_{\ell}$ such that $Q=(\Pi_{i=1}^{\ell}\pi_i)(P)$, where each $\pi_i$ (for $1\le i \le \ell$) is one of the following type: 
\begin{itemize}
    \item \emph{Replacement:} There is $j\in[1,|P|]$ and $\sigma\neq P[j]\in\Sigma$ s.\@t.\@ $\pi_i(P)[j]=\sigma$ and $\pi_i(P)[j']=P[j']$ for $j'\neq j$.  
    \item \emph{Deletion:} One has $|\pi_i(P)|=|P|-1$ and there is $j\in[1,|P|]$ s.\@t.\@ $\pi_i(P)[j']=P[j']$ for $1\le j' \le j -1$ and $\pi_i(P)[j']=P[j'+1]$ for $j\le j'\le |P|-1$.
    \item \emph{Insertion:} One has $|\pi_i(P)|=|P|+1$ and there is $j\in[1,|P|+1]$ s.\@t.\@ $\pi_i(P)[j']=P[j']$ for $1\le j' \le j-1$ and $\pi_i(P)[j']=P[j'-1]$ for $j+1 \le  j'\le |P|+1$.
\end{itemize}
\end{definition}

\begin{lemma}[\cite{DBLP:books/daglib/0020103}]
The function $d_E$ is a distance on $\Sigma^*$.
\end{lemma}

The following lemma follows immediately from Definition~\ref{def:ed}.

\begin{lemma}\label{lem:distance 1}
If $P$, $Q$ are two strings with $d_E(P,Q)=1$, then $P=\pi (Q)$ where $\pi$ is a replacement, a deletion or an insertion.
\end{lemma}

We define the main problem considered in this paper as follows:

\defproblem{$1$-Error EDSM}{A string $P$ of length $m$ and an ED string $\tilde{T}$ of length $n$ and size $N$.}{All positions $j'$ in $\tilde{T}$ such that there is at least one string $P'$ with an occurrence ending at position $j'$ in $\tilde{T}$, and with $d_E(P,P')\le 1$ (reporting version); or YES if and only if there is at least one string $P'$ with an occurrence in $\tilde{T}$, and with $d_E(P,P')\le 1$ (decision version).}

Let $P'$ be a string starting at position $j$ and ending at position $j'$ in $\tilde{T}$ with $d_E(P,P')=1$. We call this \emph{an occurrence of $P$ with $1$ error} (or a \emph{1-error occurrence}); or equivalently, we say that $P$ \emph{matches $\tilde{T}[j\dd j']$ with $1$  error}. Let $UP'_j,\ldots,P'_{j'}V$ be an alignment of $P'$ with $\tilde{T}[j\dd j']$ and $i\in[j,j']$ be an integer such that the single replacement, insertion, or deletion required to obtain $P$ from $P'=P'_j\cdot \ldots \cdot P'_{j'}$ occurs on $P'_{i}$. We then say that the alignment (and the occurrence) \emph{has the $1$ error in $\tilde{T}[i]$}. (It should be clear that for one alignment we may have multiple different $i$.)
We show the following theorem.


\begin{theorem}\label{the:main}
Given a pattern $P$ of length $m$ and an ED text $\tilde{T}$ of length $n$ and size $N$, the reporting version of
\textsc{$1$-Error EDSM} can be solved on-line in $\cO(nm^2\log m + N\log m)$ or $\cO(nm^3 + N)$ time. 
The decision version of \textsc{$1$-Error EDSM} can be solved off-line in $\cO(nm^2\sqrt{\log m} + N\log\log m)$ time.
\end{theorem}

Hamming distance, denoted by $d_H$, is a special case of edit distance in which only replacement operations are allowed (it is therefore defined for two strings of equal length). We define the following problem:

\defproblem{1-Mismatch EDSM}{A string  $P$ of length $m$ and an ED string $\tilde{T}$ of length $n$ and size $N$.}{All positions $j'$ in $\tilde{T}$ such that there is at least one string $P'$ with an occurrence ending at position $j'$ in $\tilde{T}$,  and with $d_H(P,P')\le 1$.}

An occurrence of a string $P'$ as in the problem definition is called an \emph{occurrence of $P$ with $1$ mismatch}. We call \emph{mismatch} the single position in the support of the replacement $\pi$ such that $\pi(P)=P'$. We show the following theorem.

\begin{theorem}\label{the:main2}
Given a pattern $P$ of length $m$ and an ED text $\tilde{T}$ of length $n$ and size $N$, \textsc{$1$-Mismatch EDSM} can be solved on-line in $\cO(nm^2 + N\log m)$ or $\cO(nm^3+N)$ time.
\end{theorem}

\begin{definition}
For a string $P=P[1\dd m]$, an ED string $\tilde{T}=\tilde{T}[1]\ldots \tilde{T}[n]$, a position $1\le i\le n$, and a distance on $\Sigma^*$, we define three sets:
\begin{itemize}
    \item $AP_i\subseteq [1,m]$, such that $j\in AP_i$ if and only if $P[1\dd j]$ is an \emph{active prefix} of $P$ in $\tilde{T}$ ending in the segment $\tilde{T}[i]$, that is, a prefix of $P$ which is also a suffix of a string in $\mathcal{L}(\tilde{T}[1]\ldots \tilde{T}[i])$.
    \item $AS_i\subseteq [1,m]$, such that $j\in AS_i$ if and only if $P[j\dd m]$ is an \emph{active suffix} of $P$ in $\tilde{T}$ starting in the segment $\tilde{T}[i]$, that is, a suffix of $P$ which is also a prefix of a string in $\mathcal{L}(\tilde{T}[i]\ldots \tilde{T}[n])$.
    \item $\oneAP_i\subseteq [1,m]$, such that $j\in \oneAP_i$ if and only if $P[1\dd j]$ is an \emph{active prefix with $1$ error} of $P$ in $\tilde{T}$ ending in the segment $\tilde{T}[i]$, that is, a prefix of $P$ which is also at distance at most $1$ from a suffix of a string in $\mathcal{L}(\tilde{T}[1]\ldots \tilde{T}[i])$.
\end{itemize}
For convenience we also define $AP_0=AS_{n+1}=\oneAP_0=\emptyset$.
\end{definition}

The following lemma shows that the computation of active suffixes can be easily reduced to computing the active prefixes for the reversed strings.
\begin{lemma}\label{lem:reverse for suffix}
Given a pattern $P=P[1\dd m]$ and an ED text $\tilde{T}=\tilde{T}[1\dd n]$, a suffix $P[j\dd m]$ of $P$ is an active suffix in $\tilde{T}$ starting in the segment $\tilde{T}[i]$ if and only if the prefix $P^R[1\dd m-j+1]=(P[j\dd m])^R$ of $P^R$ is an active prefix in $\tilde{T}^R$, ending in the segment $\tilde{T}^R[n-i+1]=(\tilde{T}[i])^R$. 
\end{lemma}
\begin{proof}
If $P[j\dd m]$ is a prefix of $S\in \mathcal{L}(\tilde{T}[i\dd n])$, then $P^R[1\dd m-j+1]$ is a suffix of $S^R\in\mathcal{L}(\tilde{T}[1\ldots n]^R)$.
From the definition of $\tilde{T}^R$ we have $\tilde{T}[i\dd n]^R = (\tilde{T[n]})^R\ldots (\tilde{T[i]})^R = \tilde{T}^R[1 \dd n-i+1]$, hence $S^R\in~\mathcal{L}(\tilde{T}^R[1\dd n-i+1])$.

This proves the forward direction of the lemma; the converse follows from symmetry.  
\end{proof}

The efficient computation of active prefixes was shown in~\cite{DBLP:conf/cpm/GrossiILPPRRVV17}, and constitutes the main part of the combinatorial algorithm for exact EDSM. Similarly, computing the sets $\oneAP$ plays the key role in the reporting version of our algorithm for \textsc{1-Error EDSM} (see Figure~\ref{fig:layout}). Finding active prefixes (and, by Lemma~\ref{lem:reverse for suffix}, suffixes) reduces to the following problem, formalized in~\cite{bernardini_et_al:LIPIcs:2019:10597}.

\defproblem{Active Prefixes Extension (APE)}{A string $P$ of length $m$, a bit vector $U$ of size $m$, and a set $\mathcal{S}$ of strings of total length $N$.}{A bit vector $V$ of size $m$ with $V[j]=1$ if and only if there exists $S\in\mathcal{S}$ and $i\in[1, m]$, such that $P[1\dd i]\cdot S = P[1 \dd j]$ and $U[i]=1$.}

\begin{lemma}[\cite{DBLP:conf/cpm/GrossiILPPRRVV17}]\label{lem:APE}
The APE problem for a string $P$ of length $m$ and a set $\mathcal{S}$ of strings of total length $N$ can be solved 
in $\cO(m^2+N)$ time.
\end{lemma}

Given an algorithm for the APE problem working in $f(m)+N$ time, we can find \emph{all} active prefixes for a pattern $P$ of length $m$ in an ED text $\tilde{T}=\tilde{T}[1]\ldots \tilde{T}[n]$ of size $N$ in $\cO(nf(m)+N)$ total time:

\begin{corollary}[\cite{DBLP:conf/cpm/GrossiILPPRRVV17}]\label{coro:AP}
For a pattern $P$ of length $m$ and an ED text $\tilde{T}=\tilde{T}[1]\ldots \tilde{T}[n]$ of total size $N$, computing the sets $AP_i$ for all $i\in[1,n]$ takes $\cO(nm^2+N)$ time.
\end{corollary}

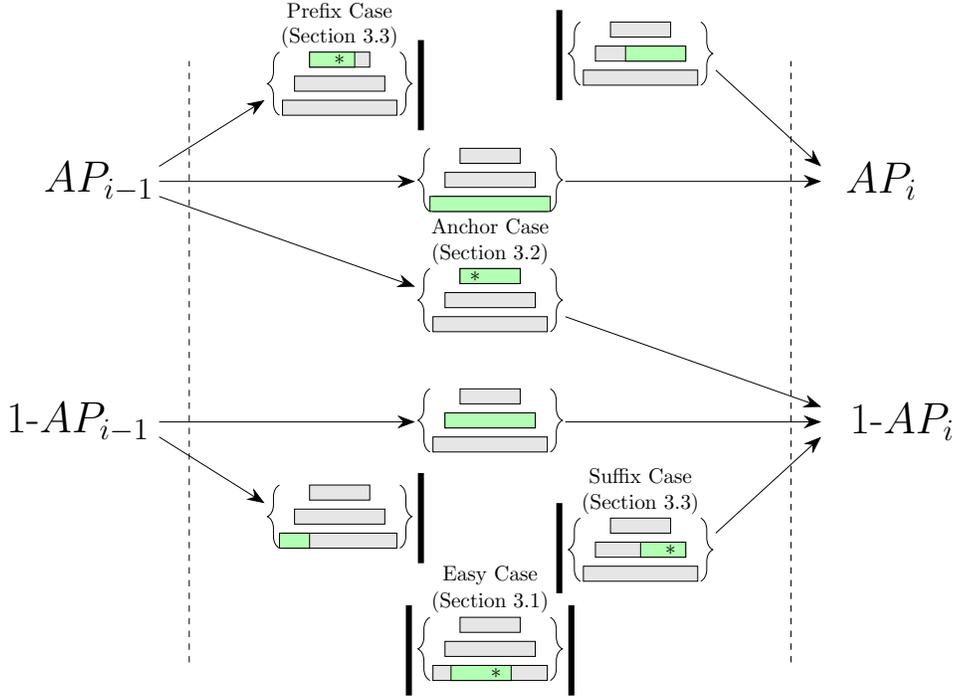
\begin{figure}
    \centering
    \scalebox{.8}{\begin{tikzpicture}[yscale=0.5,xscale=0.5,auto,node distance=0.1cm]

    \foreach \dx/\dy/\l/\lr in {5/12.2//,-5/11.2/Prefix Case/(Section \ref{sec:pref case}),0/8//,0/4/Anchor Case/(Section \ref{sec: anchor case}),0/0//,
    -5/-3.2//,5/-4.3/Suffix Case/(Section \ref{sec:pref case}),0/-7.6/Easy Case/(Section \ref{sec:easy case})}{
    \draw[snake=brace, segment amplitude=0.2cm] (\dx+2,\dy+1.1)--(\dx+2,\dy-1);
    \draw[snake=brace, segment amplitude=0.2cm] (\dx-2,\dy-1)--(\dx-2,\dy+1.1);
    \draw [fill=white!90!black] (\dx-1.5,\dy-0.2) rectangle (\dx+1.5,\dy+0.3);
    \draw [fill=white!90!black] (\dx-1,\dy+0.6) rectangle (\dx+1,\dy+1.1);
    \draw [fill=white!90!black] (\dx-1.9,\dy-1) rectangle (\dx+1.9,\dy-0.5);
    \node[] at (\dx,\dy+2.5) {\l};
    \node[] at (\dx,\dy+1.6) {\lr};
    }
    
    \node[] at (-13,8) {\huge $AP_{i-1}$};
    \node[] at (13,8) {\huge $AP_{i}$};
    
    \node[] at (-13.7,0) {\huge $\oneAP_{i-1}$};
    \node[] at (13.7,0) {\huge $\oneAP_{i}$};
    \foreach \dx in{-10,10}
    \draw[dashed] (\dx, -8)--+(0, 20);
    
    \foreach \dxa\dxb\dya\dyb in {
    -11/-7.5/8.5/10.7,
    7.5/11/11.7/8.5,
    -11/-2.5/8/8, 2.5/11/8/8,
    -11/-2.5/7.5/4.5, 2.5/11/3.5/0.5,
    -11/-2.5/0/0, 2.5/11/0/0,
    -11/-7.5/-0.5/-2.7,
    7.5/11/-3.7/-0.5
    }{
        \draw[-{Stealth[scale=1.7]}] (\dxa,\dya) -- (\dxb,\dyb);
    }
    
    \foreach \dx\dy in{
    2.3/12.2,
    -2.3/11.2,
    -2.3/-3.2,
    2.3/-4.2,
    -2.7/-7.6,2.7/-7.6
    }{
    \draw[black,line width=1mm] (\dx,\dy-1.5)--(\dx,\dy+1.5);
    }
    
    \foreach \dx/\dy/\len in {
    4.5/12/2,
    -2/7/4,
    -1.5/-0.2/3,
    -7/-4.2/1
    }{
    \draw [fill=white!70!green] (\dx,\dy) rectangle (\dx+\len,\dy+0.5);
    }
    
    \foreach \dx/\dy/\len/\starx in {
    -6/11.8/1.5/1,
    -1/4.6/2/0.5,
    5/-4.5/1.5/1,
    -1.3/-8.6/2/1.5
    }{
    \draw [fill=white!70!green] (\dx,\dy) rectangle (\dx+\len,\dy+0.5);
    \node[] at (\dx+\starx,\dy+0.25){$*$};
    }

\end{tikzpicture}}
\caption{The layout of the algorithms for computing $AP_i$, $\oneAP_i$, and reporting occurrences. The green areas correspond to the (partial) matches in $\tilde{T}[i]$, and the symbol $*$ indicates the position of an error. The vertical bold lines indicate the beginning/the end of an occurrence or a $1$-error occurrence. The cases without a label allow only exact matches and were already solved by Grossi et al.~in~\cite{DBLP:conf/cpm/GrossiILPPRRVV17}.}\label{fig:layout}
\end{figure}

As depicted in Figure~\ref{fig:layout}, the computation of active prefixes with $1$  error ($\oneAP_i$) and the reporting of occurrences with $1$  error reduce to a problem where the error can only occur in a single, fixed $\tilde{T}[i]$. In particular, this problem decomposes into 4 cases, which we formalize in the following proposition.

\begin{proposition}\label{prop:occs}
Let $\tilde{T}=\tilde{T}[1]\ldots \tilde{T}[n]$ be an ED text and $P$ be a pattern that has an occurrence with $1$ error (resp. $1$ mismatch) in $\tilde{T}$. For each alignment corresponding to such occurrence, at least one of the following is true:
\begin{description}
     \item[Easy Case:] $P$ matches $\tilde{T}[i]$ with $1$ error (resp. $1$ mismatch) for some $1\le i \le n$.
    \item [Anchor Case:] $P$ matches $\tilde{T}[j\dd j']$ with $1$ error (resp. $1$ mismatch) in $\tilde{T}[i]$ for some $1\le j < i < j'\le n$. $\tilde{T}[i]$ is called the \emph{anchor} of the alignment.
    \item[Prefix Case:] $P$ matches $\tilde{T}[j\dd i]$ with $1$ error (resp. $1$ mismatch) in $\tilde{T}[i]$ for some $1\le j < i \le n$, implying an active prefix of $P$ which is a suffix of a string in $\mathcal{L}(\tilde{T}[j\dd i-1])$.
    \item[Suffix Case:] $P$ matches $\tilde{T}[i\dd j']$ with $1$ error (resp. $1$ mismatch) in $\tilde{T}[i]$ for some $1\le i < j' \le n$, implying an active suffix of $P$ which is a prefix of a string in $\mathcal{L}(\tilde{T}[i+1\dd j'])$.
\end{description}
\end{proposition}
\begin{proof}
Suppose $P$ has a $1$-error (resp. $1$ mismatch) occurrence matching $\tilde{T}[j\dd j']$ with $1 \le j \le j' \le n$. If $j=j'$ we are in the Easy Case. Otherwise, each alignment has an error in some $\tilde{T}[i]$ for $j\le i \le j'$. If $j<i<j'$, we are in the Anchor Case; if $j<i=j'$, we are in the Prefix Case; and if $j=i<j'$, we are in the Suffix Case.\qed
\end{proof}

\section{1-Error EDSM}\label{sec:edit}

In this section, we present algorithms for finding all $1$-occurrences of $P$ given by each type of possible alignment
described by Proposition~\ref{prop:occs} (inspect Figure~\ref{fig:occs}). The Prefix and Suffix cases are analogous by Lemma~\ref{lem:reverse for suffix}; the only difference is in that, while the Suffix Case computes new $\oneAP$, the Prefix Case is used to actually report occurrences. They are jointly considered in Section~\ref{sec:pref case}.  

We follow two different procedures for the decision and reporting versions.
For the decision version, we precompute sets $AP_i$ and $AS_i$, for all $i\in[1,n]$, using Corollary~\ref{coro:AP}, and we simultaneously compute possible exact occurrences of $P$. Then we compute $1$-error occurrences of $P$ by grouping the alignments depending on the segment $i$ in which the error occurs, and using $AP_i$ and $AS_i$. For the reporting version, we consider one segment $\tilde{T}[i]$ at a time (on-line) and extend partial exact or $1$-error occurrences of $P$ to compute sets $AP_{i}$ and $\oneAP_{i}$ using just sets $AP_{i-1}$ and $\oneAP_{i-1}$ computed at the previous step. We design different procedures for the 4 cases of Proposition~\ref{prop:occs}. We can sort all letters of $P$, assign them rank values from $[1,m]$, and construct a perfect hash table over these letters supporting $\cO(1)$-time look-up queries in $\cO(m\log m)$ time~\cite{DBLP:conf/icalp/Ruzic08}. Any letter of $\tilde{T}$ not occurring in $P$ can be replaced by the same special letter in $\cO(1)$ time. In the rest we thus assume that the input strings are over $[1,m+1]$.

Two problems from computational geometry have a key role in our solutions. We assume the word RAM model with coordinates on the integer grid $[1,n]^d=\{1,2,\ldots,n\}^d$. In the \emph{2d rectangle emptiness} problem, we are given a set $\mathcal{P}$ of $n$ points to be preprocessed, so that when one gives an axis-aligned rectangle as a query, we report YES if and only if the rectangle contains a point from $\mathcal{P}$. In the ``dual'' \emph{2d rectangle stabbing} problem, we are given a set $\mathcal{R}$ of $n$ axis-aligned rectangles to be preprocessed, so that when one gives a point as a query, we report YES if and only if there exists a rectangle from $\mathcal{R}$ containing the point.

\begin{lemma}[\cite{DBLP:conf/compgeom/ChanLP11,DBLP:conf/esa/Gao0N20}]\label{lem:2demptiness}
After $\cO(n \sqrt{\log n})$-time preprocessing, we can answer 2d rectangle emptiness queries in $\cO(\log \log n)$ time.
\end{lemma}

\begin{lemma}[\cite{DBLP:journals/siamcomp/Chazelle88,DBLP:journals/siamcomp/ShiJ05}]\label{lem:2dstabbing}
After $\cO(n \log n)$-time preprocessing, we can answer 2d rectangle stabbing queries in $\cO(\log n)$ time.
\end{lemma}

In Section~\ref{sec:CG}, we note that the 2d rectangle stabbing instances
arising from \textsc{1-Error EDSM} have a special structure.
We show how to solve them efficiently thus shaving logarithmic factors from the time complexity.

\subsection{Easy Case}\label{sec:easy case}

The Easy Case can be reduced to approximate string matching with at most 1 error ($1$-SM):

\defproblem{$1$-SM}{A string  $P$ of length $m$ and a string $T$ of length $n$.}{All positions $j$ in $T$ such that there is at least one string $P'$ ending at position $j$ in $T$ with $d_E(P,P')\le 1$.}

We have the following well-known results.

\begin{lemma}[\cite{DBLP:journals/jcss/LandauV88,DBLP:journals/siamcomp/ColeH02}]\label{lem:kPM}
Given a pattern $P$ of length $m$, a text $T$ of length $n$, and an integer $k>0$, all positions $j$ in $T$ such that the edit distance of $T[i\dd j]$ and $P$, for some position $i\leq j$ on $T$, is at most $k$, can be found in $\cO(kn)$ time or in $\cO(\frac{nk^4}{m}+n)$ time.\footnote{Charalampopoulos et al.~have announced an improvement on the exponent of $k$ from 4 to 3.5; specifically they presented an $\cO(\frac{nk^{3.5}\sqrt{\log m\log k}}{m}+n)$-time algorithm~\cite{DBLP:journals/corr/abs-2204-03087}.} In particular, $1$-SM can be solved in $\cO(n)$ time.
\end{lemma}

We find occurrences of $P$ with at most $1$  error that are in the Easy Case for segment $\tilde{T}[i]$ in the following way: we apply Lemma~\ref{lem:kPM} for $k=1$ and every string of  $\tilde{T}[i]$ whose length is at least $m-1$ (any shorter string is clearly not relevant for this case) as text. If, for any of those strings, we find an occurrence of $P$, we report an occurrence at position $i$ (inspect Figure~\ref{fig:easy}). The time for processing a segment $\tilde{T}[i]$ is $\cO(N_i)$, where $N_i$ is the total length of all the strings in $\tilde{T}[i]$.

\begin{figure}[t]
    \begin{subfigure}{\textwidth}
    \centering
    \includegraphics[width=6cm]{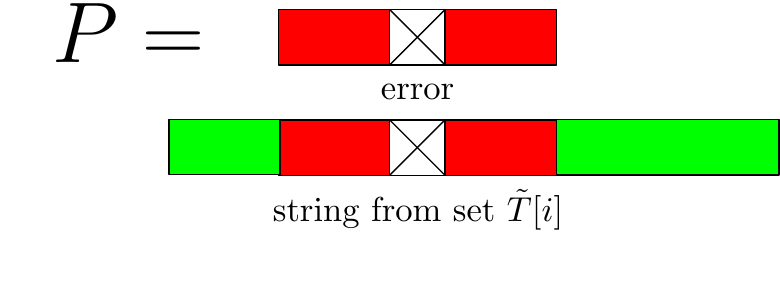}
    \caption{Easy Case: $j = i = j'$.}
    \label{fig:easy}
    \end{subfigure}
    
    \begin{subfigure}{\textwidth}
    \medskip
    \centering
    \includegraphics[width=9cm]{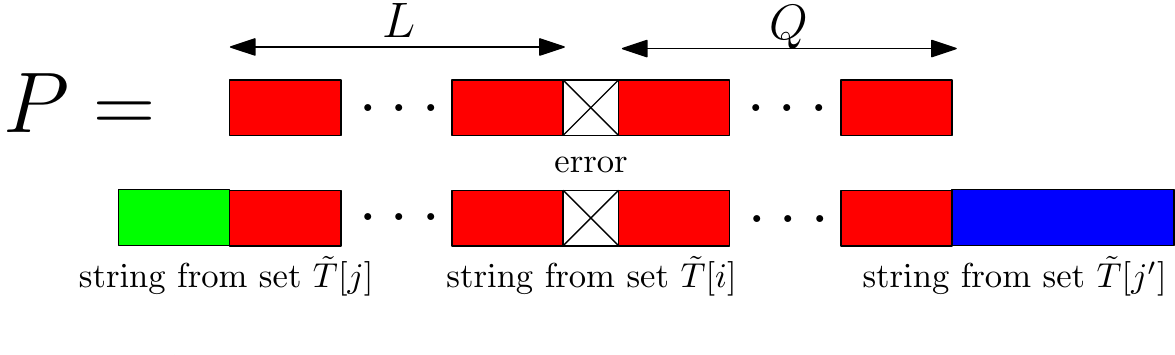}
    \caption{Anchor Case: $j\neq i,i\neq j'$.}
    \label{fig:anchor}
    \end{subfigure}

    \begin{subfigure}{\textwidth}
    \medskip
    \centering
    \includegraphics[width=7cm]{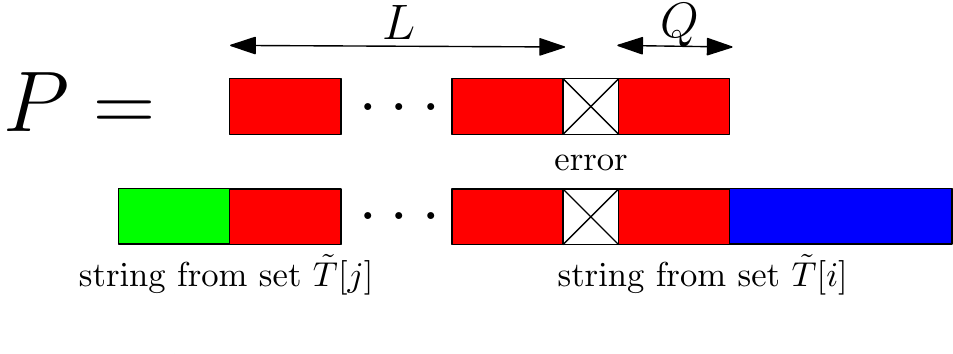}
    \caption{Prefix Case: $j\neq i,i=j'$.}
    \label{fig:prefix}
    \end{subfigure}
    
    \begin{subfigure}{\textwidth}
    \medskip
    \centering
    \includegraphics[width=7cm]{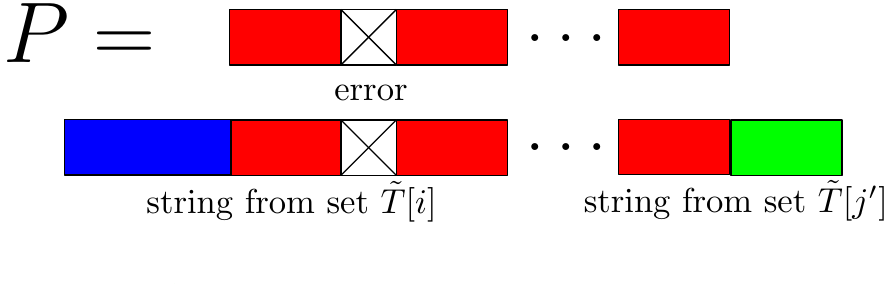}
    \caption{Suffix Case: $i=j,i\neq j'$.}
    \label{fig:suffix}
    \end{subfigure}
    
    \caption{Possible alignments of $1$-error occurrences of $P$ in $\tilde{T}$. Each occurrence starts at segment $\tilde{T}[j]$, ends at $\tilde{T}[j']$, and the error occurs at $\tilde{T}[i]$.}
    \label{fig:occs}
\end{figure}

\subsection{Anchor Case}\label{sec: anchor case}

Let $\tilde{T}$ be an ED text and $P$ be a pattern with a $1$-error occurrence and an alignment in the Anchor Case with anchor $\tilde{T}[i]$. Further let $L=P[1\dd \ell]S'$ and $Q=S''P[q\dd m]$ be a prefix and a suffix of $P$, respectively, for some $\ell\in AP_{i-1}, q\in AS_{i+1}$, where $S',S''$ are a prefix and a suffix of some $S\in \tilde{T}[i]$, respectively (strings $S',S''$ can be empty). By Lemma~\ref{lem:distance 1}, a pair $L,Q$ gives a $1$-error occurrence of $P$ if one of the following holds:

\begin{description}
  \item[1 mismatch:] $|L|+|Q|+1=m$ \emph{and} $|S'|+|S''|+1=|S|$ (inspect Figure~\ref{fig:anchor}). 
 \item[1 deletion in $P$:] $|L|+|Q|=m-1$ \emph{and} $|S'|+|S''|=|S|$.
 \item[1 insertion in $P$:] $|L|+|Q|=m$ \emph{and} $|S'|+|S''|+1=|S|$.
\end{description}


We show how to find such pairs with the use of a geometric approach. For convenience, we only present the Hamming distance (1 mismatch) case. The other cases are handled similarly.


Let $\lambda\in AP_{i-1}$ be the length of an active prefix, and let $\rho$ be the length of an active suffix, that is, $m-\rho+1\in AS_{i+1}$. Note that $AP_{i-1}$ and $AS_{i+1}$ can be precomputed, for all $i$, in $\cO(nm^2+N)$ total time by means of Corollary~\ref{coro:AP}. (In particular, $AS_{i+1}$ is required only for the decision version; for the reporting version, we explain later on how to avoid the precomputation of $AS_{i+1}$ to obtain an on-line algorithm.) We will exhaustively consider all pairs $(\lambda,\rho)$ such that $\lambda+\rho<m$. Clearly, there are $\cO(m^2)$ such pairs.

Consider the length $\mu = m - (\lambda+\rho) > 0$ of the substring of $P$ still to be matched for some prefix and suffix of $P$ of lengths $(\lambda, \rho)$, respectively.
We group together all pairs  $(\lambda,\rho)$ such that $m - (\lambda+\rho)=\mu$ by sorting them in $\cO(m^2)$ time.
We construct, for each such group $\mu$, the compacted trie $T_\mu$ of the fragments $P[\lambda+1\dd m-\rho]$, for all $(\lambda,\rho)$ such that $m - (\lambda+\rho)=\mu$, and analogously the compacted trie $T^R_\mu$ of all fragments $P^R[\rho+1\dd m-\lambda]$. 
For each group $\mu$, this takes $\cO(m)$ time~\cite{DBLP:journals/tcs/NaAIP03}. We enhance all nodes with a perfect hash table in $\cO(m)$ total time to access edges by the first letter of their label in $\cO(1)$ time~\cite{DBLP:journals/jacm/FredmanKS84}.

We also group all strings in segment $\tilde{T}[i]$ of length less than $m$ by their length $\mu$.
The group for length $\mu$ is denoted by $G_\mu$.
This takes $\cO(N_i)$ time.
Clearly, the strings in $G_\mu$ are the only candidates to extend pairs $(\lambda,\rho)$ such that $ m - (\lambda+\rho)=\mu$.
Note that the mismatch can be at any position of any string of $G_\mu$: its position determines a prefix $S'$ of length $h$ and a suffix $S''$ of length $k$ of the same string $S$, with $h+k=\mu-1$, that must match a prefix and a suffix of $P[\lambda+1\dd m-\rho]$, respectively.
We will consider all such pairs of positions $(h,k)$ whose sum is $\mu-1$ (intuitively, the minus one is for the mismatch).
This guarantees that $L=P[1\dd\lambda]S'$ and $Q=S''P[m-\rho+1\dd m]$ are such that \emph{$|L|+|Q|+1=m$}.
The pairs are $(0,\mu-1),(1,\mu-2),\ldots,(\mu-1,0)$.
This guarantees that $L$ and $Q$ are \emph{one position apart} ($|S'|+|S''|+1=|S|$).

The number of these pairs is $\cO(\mu)=\cO(m)$. Consider one such pair $(h,k)$ and
a string $S\in G_\mu$. We treat every such string $S$ separately.
We spell $S[1\dd h]$ in $T_\mu$.
If the whole $S[1\dd h]$ is successfully spelled ending at a node $u$, this implies that all the fragments of $P$ corresponding to nodes descending from $u$ share $S[1\dd h]$ as a prefix.
We also spell $S^R[1\dd k]$ in $T^R_\mu$. If the whole of $S^R[1\dd k]$ is successfully spelled ending at a node $v$, then all the fragments of $P$ corresponding to nodes descending from $v$ share $(S^R[1\dd k])^R$ as a suffix.
Nodes $u$ and $v$ identify an interval of leaves in $T_\mu$ and $T^R_\mu$, respectively. 
We need to check if these intervals both contain a leaf corresponding to the same fragment of $P$.  If they do, then we obtain an occurrence of $P$ with $1$ mismatch (see Figure~\ref{fig:rectangles}). We now have two different ways to proceed, depending on whether we need to solve the off-line decision version or the on-line reporting version.

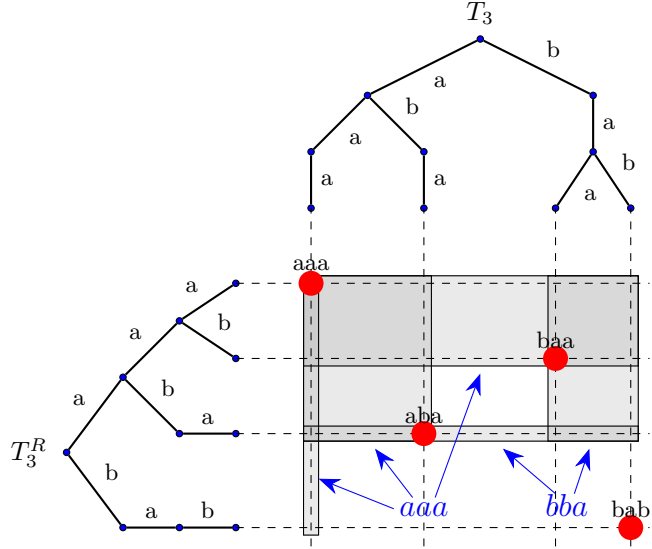
\begin{figure}
    \centering
    \begin{tikzpicture}[yscale=0.5,xscale=0.5,auto,node distance=0.1cm]

\tikzstyle{dot}=[inner sep=0.03cm, circle, draw]

\node[] at (11,11.7) {$T_{3}$};
\node[dot,fill=blue] (t) at (11,11) {};
\foreach \name/\dx/\parent/\l in {
  l/-3/t/a,
  r/+3/t/b,
  ll/-1.5/l/a,
  lr/+1.5/l/b,
  rl/0/r/a,
  lll/0/ll/a,
  lrl/0/lr/a,
  rll/-1/rl/a,
  rlr/+1/rl/b
} {
  \node[dot,fill=blue] (\name) at ($(\parent)+(\dx,-1.5)$) {};
  \draw[black,thick] (\parent)--(\name) node[midway,auto] {\small \l};
}

\node[] at (-1,0) {$T^R_3$};
\node[dot,fill=blue] (t) at (0,0) {};
\foreach \name/\dy/\parent/\l in {
  l/+2/t/a,
  r/-2/t/b,
  ll/+1.5/l/a,
  lr/-1.5/l/b,
  rl/0/r/a,
  lll/+1/ll/a,
  llr/-1/ll/b,
  lrl/0/lr/a,
  rll/0/rl/b
} {
  \node[dot,fill=blue] (\name) at ($(\parent)+(1.5,\dy)$) {};
  \draw[black,thick] (\parent)--(\name) node[midway,auto] {\small \l};
}

\foreach \dxa/\dxb/\dya/\dyb/\l in {
-4.5/4/2.5/4.5/$\_aa$, 
-4.5/-1.5/0.5/4.5/$a\_a$, 
-4.5/-4.5/-2/4.5/$aa\_$, 
-4.5/4/0.5/0.5/$\_ba$, 
2/4/0.5/4.5/$b\_a$ 
}{
\draw[fill=black!40!white,opacity=0.2] (\dxa+10.8,\dya-0.2) rectangle (\dxb+11.2,\dyb+0.2);
}

\foreach \dx in{-4.5,-1.5,2,4}
\draw[dashed] (\dx+11, -2.5)--+(0, 9);

\foreach \dy in{-2,0.5,2.5,4.5}
\draw[dashed] (4.5, \dy)--+(11, 0);

\foreach \dxa/\dxb/\dya/\dyb/\l in {
-4.5/4/2.5/4.5/$\_aa$, 
-4.5/-1.5/0.5/4.5/$a\_a$, 
-4.5/-4.5/-2/4.5/$aa\_$, 
-4.5/4/0.5/0.5/$\_ba$, 
2/4/0.5/4.5/$b\_a$ 
}{
\draw[draw=black] (\dxa+10.8,\dya-0.2) rectangle (\dxb+11.2,\dyb+0.2);
}

\foreach \dx/\dy/\l in {-4.5/+4.5/aaa,-1.5/0.5/aba,2/2.5/baa,4/-2/bab}{
\node[circle, fill=red] at (\dx+11,\dy) {};
\node[] at (\dx+11,\dy+0.5) {\small \l};
}

\node[] at (9.5,-1.5) {\large \textcolor{blue}{$aaa$}};
\node[] at (13.3,-1.4) {\large \textcolor{blue}{$bba$}};

    \foreach \dxa\dxb\dya\dyb in {
    8.7/6.7/-1.5/-0.7,
    9.2/8.2/-1.2/0.1,
    9.8/11/-1.2/2.1,
    13.0/11.6/-1.1/0.1,
    13.5/13.9/-1.1/0.1
    }{
        \draw[draw=blue,-{Stealth[scale=1.7]}] (\dxa,\dya) -- (\dxb,\dyb);
    }

\end{tikzpicture}
\caption{An example of points and rectangles (solid shapes) for the decision version of the Anchor Case with $1$  mismatch. Here $P=bbaaaabababb$, $AP_{i-1}=\{1,2,4,7,8,9\},AS_{i+1}=\{5,6,9,11,12\}$, $\mu=3$, and $\tilde{T}[i]=\{aaa,bba\}$.
$T_3$ and $T^R_3$ are built for 4 strings: $P[2\dd4]=baa,P[3\dd5]=aaa,P[8\dd10]=aba,P[9\dd11]=bab$; the 5 rectangles correspond to pairs $(\varepsilon,aa),(a,a),(aa,\varepsilon),(\varepsilon,ab),(b,a)$, namely, the pairs of prefixes and reversed suffixes of $aaa$ and $bba$ (rectangle $(bb,\varepsilon)$ does not exist as $T_3$ contains no node $bb$).
}\label{fig:rectangles}
\end{figure}

\subparagraph{Decision Version} Let us recall that $T_\mu,T^R_\mu$ by construction are ordered based on lexicographic ranks. For every pair $(T_\mu,T^R_\mu)$, we construct a data structure for 2d rectangle emptiness queries on the grid $[1,\ell]^2$, where $\ell$ is the number of leaves of $T_\mu$ (and of $T_\mu^R$), for the set of points $(x,y)$ such that $x$ is the lexicographic rank of the leaf of $T_\mu$ representing $P[\lambda+1\dd m-\rho]$ and $y$ is the rank of the leaf of $T_\mu^R$ representing $P^R[\rho+1\dd m-\lambda]$ for the same pair $(\lambda,\rho)$. This denotes that the two leaves correspond to \emph{the same fragment} of $P$. For every $(T_\mu,T^R_\mu)$, this preprocessing takes $\cO(m\sqrt{\log m})$ time by Lemma~\ref{lem:2demptiness}, since $\ell$ is $\cO(\mu)=\cO(m)$. For all $\mu$ groups (they are at most $m$), the whole preprocessing thus takes $\cO(m^2\sqrt{\log m})$ time. 

We then ask 2d range emptiness queries that take $\cO(\log\log m)$ time each by Lemma~\ref{lem:2demptiness}. Note that all rectangles for $S$ can be collected in $\cO(|S|)=\cO(\mu)$ time by spelling $S$ through $T_\mu$ and $S^R$ through $T^R_\mu$, one letter at a time. Thus the total time for processing all $G_\mu$ groups of segment $i$ is $\cO(m^2\sqrt{\log m}+N_i\log\log m)$.
    If any of the queried ranges turns out to be non-empty, then $P'$ such that $d_H(P,P')\le 1$ appears in $\mathcal{L}(\tilde{T})$ with anchor in $\tilde{T}[i]$; we do not have sufficient information to output its ending position however.
    
\subparagraph{Reporting Version} For this version, we do the dual.
    We construct a data structure for 2d rectangle stabbing queries on the grid $[1,\ell]^2$ for the set of rectangles collected for all strings $S\in G_\mu$. 
    By Lemma~\ref{lem:2dstabbing}, for all $\mu$ groups, the whole preprocessing thus takes $\cO(N_i\log N_i)$ time.
    
    For every $(T_\mu,T^R_\mu)$, we then ask the following queries: $(x,y)$ is queried if and only if $x$ is the rank of a leaf representing $P[\lambda+1\dd m-\rho]$ and $y$ is the rank of a leaf representing $P^R[\rho+1\dd m-\lambda]$. 
    For every $(T_\mu,T^R_\mu)$, this takes $\cO(m\log N_i)$ time by Lemma~\ref{lem:2dstabbing} and by the fact that for each group $G_\mu$ there are $\cO(m)$ pairs $(\lambda,\rho)$ such that $m-(\lambda+\rho)=\mu$. For all groups $G_\mu$  (they are at most $m$), all the queries thus take $\cO(m^2\log N_i)$ time. Thus the total time for processing all $G_\mu$ groups of segment $i$ is $\cO((m^2+N_i)\log N_i)$.
    
    We are not done yet. 
    By performing the above algorithm for active prefixes and active suffixes, we find out which pairs can be completed to a full occurrence of $P$ with at most $1$  error. This information is not sufficient to compute where such an occurrence ends (and storing additional information together with the active suffixes may prove costly).
    To overcome this, we use some ideas from the decision algorithm, appropriately modified to preserve the on-line nature of the reporting algorithm.
    Instead of iterating $\rho$ over the lengths of precomputed active suffixes, we iterate it over \emph{all} possible lengths in $[0, m]$ (including $0$ because we may want to include $m$ in $\oneAP_i$).
    A suffix of $P$ of length $\rho$ completes a partial occurrence computed up to segment $i$ exactly when $m-\rho\in \oneAP_i$ (a pair $x\in \oneAP_i, x+1\in AS_{i+1}$ corresponds to an occurrence).
    We thus use the reporting algorithm to compute the part of $\oneAP_i$ coming from the extension of $AP_{i-1}$ (see Figure~\ref{fig:layout}), and defer the reporting to the no-error version of the Prefix Case for the right $j'$; which was solved by Grossi et al.~\cite{DBLP:conf/cpm/GrossiILPPRRVV17} in linear time.

\subsection{Prefix Case}\label{sec:pref case}

Let $\tilde{T}$ be an ED text and $P$ be a pattern with a $1$-error occurrence and an alignment in the Prefix Case with active prefix ending at $\tilde{T}[i-1]$.
Let $L=P[1\dd \ell]S'$, with $\ell\in AP_{i-1}$, be a prefix of $P$ that is extended in $\tilde{T}[i]$ by $S'$; and $Q$ be a suffix of $P$ occurring in some string of $\tilde{T}[i]$ (strings $S',Q$ can be empty). By Lemma~\ref{lem:distance 1}, we have $3$ possibilities for any alignment of a $1$-error occurrence of $P$ in the Prefix Case:
\begin{description}
    \item[1 mismatch:] $|L|+|Q|+1=m$, $S'$ is a prefix of the same string in which $Q$ occurs, and they are one position apart (inspect Figure~\ref{fig:prefix}). 
    \item[1 deletion in $P$:] $|L|+|Q|=m-1$, $S'$ is a prefix of the same string in which $Q$ occurs, and they are consecutive.
    \item[1 insertion in $P$:] $|L|+|Q|=m$, $S'$ is a prefix of the same string in which $Q$ occurs, and they are one position apart. 
\end{description}

For convenience, we only present the method for Hamming distance (1 mismatch). The other possibilities are handled similarly.

The techniques are similar to those for the Anchor Case (Section~\ref{sec: anchor case}).
We group the prefixes of all strings in $\tilde{T}[i]$ according to their length $\mu\in[1,m)$. The total number of these prefixes is $\cO(N_i)$.
The group for length $\mu$ is denoted by $G_\mu$.
We construct the compacted trie $T_{G_\mu}$ of the strings in $G_\mu$, and the compacted trie $T^R_{G_\mu}$ of the reversed strings in $G_\mu$.
This can be done in $\cO(N_i)$ total time for all compacted tries.
To achieve this, we employ the following lemma by Charalampopoulos et al.~\cite{DBLP:journals/jea/Charalampopoulos20}. (Recall that we have already sorted all letters of $P$. In what follows, we assume that $N_i\geq m$; if this is not the case, we can sort all letters of $\tilde{T}[i]$ in $\cO(m + N_i)$ time.)
\begin{lemma}[\cite{DBLP:journals/jea/Charalampopoulos20}]\label{lem:sort}
Let $X$ be a string of length $n$ over an integer alphabet of size $n^{\cO(1)}$. Let $I$ be a collection of intervals $[i,j]\subseteq[1,n]$. We can lexicographically sort the substrings $X[i\dd j]$ of $X$, for all intervals $[i,j]\in I$, in $\cO(n+|I|)$ time.
\end{lemma}

We concatenate all the strings of $\tilde{T}[i]$ to obtain a single string $X$ of length $N_i$, to which we apply, for each $\mu$, Lemma~\ref{lem:sort}, with a set $I$ consisting of the intervals over $X$ corresponding to the strings in $G_\mu$.
By sorting, in this way, all strings in $G_\mu$ (for all $\mu$),
and by constructing~\cite{DBLP:conf/focs/Farach97} and preprocessing~\cite{DBLP:conf/latin/BenderF00} the generalized suffix tree of the strings in $\tilde{T}[i]$ in $\cO(N_i)$ time to support answering lowest common ancestor (LCA) queries in $\cO(1)$ time, we can construct all $T_{G_\mu}$ in $\cO(N_i)$ total time. We handle $T^R_{G_\mu}$, for all $\mu$, analogously. Similar to the Anchor Case we enhance all nodes with a perfect hash table within the same complexities~\cite{DBLP:journals/jacm/FredmanKS84}.

In contrast to the Anchor Case, we now
only consider the set $AP_{i-1}$: namely, we do not consider $AS_{i+1}$. Let $\lambda\in AP_{i-1}$ be the length of an active prefix. We treat every such element separately, and they are $\cO(m)$ in total. Let $\mu=m-\lambda>0$ and consider the group $G_\mu$ whose strings are all of length $\mu$.
The mismatch being at position $h+1$ in one such string $S$ determines a prefix $S'$ of $S$ of length $h$ that must extend the active prefix of $P$ of length $\lambda$, and a fragment $Q$ of $S$ of length $k=\mu-h-1$ that must match a suffix of $P$.
We will consider all such pairs $(h,k)$ whose sum is $\mu-1$.
The pairs are again $(0,\mu-1),(1,\mu-2),\ldots,(\mu-1,0)$, and there are clearly $\cO(\mu)=\cO(m)$ of them.

Consider $(h,k)$ as one such pair.
We spell $P[\lambda+1\dd \lambda+h]$ in $T_{G_\mu}$. 
If the whole $P[\lambda+1\dd \lambda+h]$ is spelled successfully, this implies an interval of leaves of $T_{G_\mu}$ corresponding to strings from $\tilde{T}[i]$ that share $P[\lambda+1\dd \lambda+h]$ as a prefix.
We spell $P^R[1\dd k]$ in $T^R_{G_\mu}$. If the whole $P^R[1\dd k]$ is spelled successfully, this implies an interval of leaves of $T_{G_\mu}^R$ corresponding to strings from $\tilde{T}[i]$ that have the same fragment $(P^R[1\dd k])^R$.
These two intervals form a rectangle in the grid implied by the leaves of $T_{G_\mu}$ and $T^R_{G_\mu}$.
We need to check if these intervals both contain a leaf corresponding to the same prefix of length $\mu$ of a string in $\tilde{T}[i]$. If they do, then we have obtained an occurrence with $1$ mismatch in $\tilde{T}[i]$. 

To do this we construct, for every $(T_{G_\mu},T^R_{G_\mu})$, a 2d range data structure for the set of points $(x,y)$ such that $x$ is the rank of a leaf of $T_{G_\mu}$, $y$ is the rank of a leaf of $T^R_{G_\mu}$, and the two leaves correspond to \emph{the same prefix} of length $\mu$ of a string in $\tilde{T}[i]$. For every $(T_{G_\mu},T^R_{G_\mu})$, this takes $\cO(|G_\mu|\sqrt{\log |G_\mu|})$ time by Lemma~\ref{lem:2demptiness}.
For all $G_\mu$ groups, the whole preprocessing takes $\cO(N_i\sqrt{\log N_i})$ time.

We then ask 2d range emptiness queries each taking $\cO(\log\log |G_\mu|)$ time by Lemma~\ref{lem:2demptiness}.
Note that all rectangles for $\lambda$ can be collected in $\cO(m)$ time by spelling $P[\lambda+1\dd \lambda+\mu-1]$ through $T_{G_\mu}$ and $P^R[1\dd \mu-1]$ through $T^R_{G_\mu}$, one letter at a time.
This gives a total of $\cO(m^2\log\log N_i+N_i\sqrt{\log N_i})$ time for processing all $G_\mu$ groups of $\tilde{T}[i]$, because $\sum_{\mu}|G_\mu|\leq N_i$.

To solve the Suffix Case (compute active prefixes with $1$ error starting in $\tilde{T}[i]$) we employ the mirror version of the algorithm, but iterating $\lambda$ over the whole $[0,m]$ instead of $AS_{i+1}$ (like in the reporting version of the Anchor Case). 

\subsection{Shaving Logs using Special Cases of Geometric Problems}\label{sec:CG}

\subsubsection{Anchor Case: Simple 2d Rectangle Stabbing}

\begin{lemma}\label{lem:m^3}
We can solve the Anchor Case (i.e., extend $AP_{i-1}$ into $\oneAP_i$) in $\cO(m^3+N_i)$ time.
\end{lemma}
\begin{proof}
By Lemma~\ref{lem:2dstabbing}, 2d rectangle stabbing queries can be answered in $\cO(\log n)$ time using $\cO(n\log n)$ space after $\cO(n \log n)$-time preprocessing.

Notice that in the case of the 2d rectangle stabbing used in Section~\ref{sec: anchor case} the rectangles and points are all in a predefined $[1,m]\times[1,m]$ grid. In such a case we can also use an easy folklore data structure of size $\cO(m^2)$, which after an $\cO(m^2+|\text{rectangles}|)$-time preprocessing answers such queries in $\cO(1)$ time.

Namely, the data structure consists of a $[1,m+1]^2$ grid $\Gamma$ (a 2d-array of integers) in which for every rectangle $[u,v]\times[w,x]$ we add $1$ to $\Gamma[u][w]$ and $\Gamma[v+1][x+1]$ and $-1$ to $\Gamma[u][x+1]$ and $\Gamma[v+1][w]$. Then we modify $\Gamma$ to contain the 2d prefix sums of its original values (we first compute prefix sums of each row, and then prefix sums of each column of the result). After these modifications, $\Gamma[x][y]$ stores the number of rectangles containing point $(x,y)$, and hence after $\cO(m^2+|\text{rectangles}|)$-time preprocessing we can answer 2d rectangle stabbing queries in $\cO(1)$ time.

In our case we have a total of $\cO(m)$ such grid structures, each of $\cO(m^2)$ size, and ask $\cO(m^2)$ queries, and hence obtain an $\cO(m^3+N_i)$-time and $\cO(m^2)$-space solution for computing $\oneAP_i$ from $AP_{i-1}$.
\end{proof}

\subsubsection{Prefix Case: a Special Case of 2d Rectangle Stabbing}\label{subsec:prefix case}


Inspect the example of Figure~\ref{fig:rectangles} for the Anchor Case. Note that the groups of rectangles for each string have the special property of being composed of \emph{nested intervals}: for each dimension,  the interval corresponding to a given node is included in the one corresponding to any of its ancestors. Thus for the Prefix Case, where we only spell fragments of the same string $P$ in both compacted tries,
we consider the following special case of off-line 2d rectangle stabbing.

\begin{lemma}\label{lem:emptiness-offline}
Let $p_1,\ldots,p_h$ and $q_1,\ldots,q_h$ be two permutations of $[1,h]$.
We denote by $\Pi$ the set of $h$ points $(p_1,q_1),(p_2,q_2),\ldots,(p_h,q_h)$ on $[1,h]^2$.

Further let $R$ be a collection of $r$ axis-aligned rectangles $([u_1,v_1],[w_1,x_1]),\ldots,([u_r,v_r],[w_r,x_r])$, such that 
$$[u_r,v_r] \subseteq [u_{r-1},v_{r-1}]\subseteq \cdots \subseteq [u_1,v_1]$$
and
$$[w_1,x_1] \subseteq [w_{2},x_{2}]\subseteq \cdots \subseteq [w_r,x_r].$$
Then we can find out, for every  point from $\Pi$, if it stabs any rectangle from $R$ in $\cO(h+r)$ total time.
\end{lemma}

\begin{proof}
Let $H$ be a bit vector consisting of $h$ bits, initially all set to zero.
We process one rectangle at a time. We start with $([u_1,v_1],[w_1,x_1])$.
We set $H[p]=1$ if and only if $(p,q)\in \Pi$ for $p\in [u_1,v_1]$ and any $q$.  We collect all $p$ such that $(p,q)\in \Pi$ and $q\in [w_1,x_1]$, and then search for these $p$ in $H$: if for any $p$, $H[p]=1$, then the answer is positive for $p$. Otherwise, we remove from $H$ every $p$ such that $p\in[u_1,v_1]$ and $p\notin[u_2,v_2]$ by setting $H[p]=0$. We proceed by collecting all $p$ such that $(p,q)\in \Pi$, $q\in [w_2,x_2]$ and $q\notin [w_1,x_1]$, and then search for them in $H$: if for any $p$, $H[p]=1$, then the answer is positive for $p$. We repeat this until $H$ is empty or until there are no other rectangles to process.

The whole procedure takes $\cO(h+r)$ time, because we set at most $h$ bits on in $H$, we set at most $h$ bits back off in $H$, we search for at most $h$ points in $H$, and
then we process $r$ rectangles.
\end{proof}



\begin{lemma}\label{lem:prefix}
We can solve the Prefix (resp.~Suffix) Case, that is, report $1$-error occurrences ending in $\tilde{T}[i]$ (resp.~compute active prefixes with $1$  error starting in $\tilde{T}[i]$)  in $\cO(m^2+N_i)$ time.
\end{lemma}
\begin{proof}
We employ Lemma~\ref{lem:emptiness-offline} to get rid of the 2d range data structure. The key is that for every length-$\mu$ suffix $P[\lambda+1\dd m]$ of the pattern we can afford to pay $\cO(\mu+|G_\mu|)$ time plus the time to construct $T_{G_\mu}$ and $T^R_{G_\mu}$ for set $G_\mu$. 
Because the grid is $[1,|G_\mu|]^2$, we exploit the fact that the intervals found by spelling $P[\lambda+1\dd \lambda+\mu-1]$ through $T_{G_\mu}$ and $P^R[1\dd \mu-1]$ through $T^R_{G_\mu}$, one letter at a time, are subset of each other, and querying $\mu$ such rectangles is done in $\cO(\mu + |G_\mu|)$ time by employing Lemma~\ref{lem:emptiness-offline}. Since we process at most $m$ distinct length-$\mu$ suffixes of $P$, the total time is $\cO(m^2+N_i)$, because $\sum_{\mu}|G_\mu|\leq N_i$.
\end{proof}


\subsection{Wrapping-up}\label{sec:wrapping up}

To obtain Theorem~\ref{the:main} for the decision version of the problem we first compute $AP_i$ and $AS_i$, for all $i\in[1,n]$, in $\cO(nm^2+N)$ total time (Corollary~\ref{coro:AP}). We then compute all the occurrences in the Easy Cases using $\cO(N)$ time in total (Section~\ref{sec:easy case}); and we finally compute all the occurrences in the Prefix and Suffix Cases in $\sum_{i}\cO(m^2+N_i)=\cO(nm^2+N)$ total time (Lemma~\ref{lem:prefix}).

Now, to solve the decision version of the problem, we solve the Anchor Cases with the use of the precomputed $AP_{i-1}$ and $AS_{i+1}$ for each $i\in[2,n-1]$ in $\cO(m^2\sqrt{\log m}+N_i\log\log m)$ time (Section~\ref{sec: anchor case}), which gives $\cO(nm^2\sqrt{\log m}+N\log\log m)$ total time for the whole algorithm.

For the reporting version we proceed differently to obtain an on-line algorithm; note that this is possible because we can proceed without $AS_i$ (see Figure~\ref{fig:layout}). We thus consider one segment $\tilde{T}[i]$ at the time, for each $i\in[1,n]$, and do the following.
 We compute $\oneAP_i$,  as the union of three sets obtained from:
    \begin{itemize}
        \item The Suffix Case for $\tilde{T}[i]$, computed in $\cO(m^2+N_i)$ time (Lemma~\ref{lem:prefix}).
        \item Standard APE with $\oneAP_{i-1}$ as the input bit vector, computed in $\cO(m^2+N_i)$ time (Lemma~\ref{lem:APE}).
        \item Anchor Case computed from $AP_{i-1}$ in $\cO((m^2+ N_i)\log N_i)$ (Section~\ref{sec: anchor case}) or $\cO(m^3+N_i)$ time (Lemma~\ref{lem:m^3}).
    \end{itemize}
    If $N_i\ge m^3$, the algorithm of Lemma~\ref{lem:m^3} works in the optimal $\cO(m^3+N_i)=\cO(N_i)$ time, hence we can assume that the $\cO((m^2+N_i)\log N_i)$-time algorithm is only used when $N_i\le m^3$, and thus it runs in $\cO((m^2+N_i)\log m)$ time.
   Therefore over all $i$ the computations require $\cO((nm^2+N)\log m)$ or $\cO(nm^3+N)$ total time. For every segment $i$ we can also check whether an active prefix from $\oneAP_{i-1}$ or from $AP_{i-1}$ can be completed to a full match in $\tilde{T}[i]$ using the algorithms of Grossi et al.~from~\cite{DBLP:conf/cpm/GrossiILPPRRVV17} and Prefix Case, respectively, in $\cO(m^2+N_i)$ extra time.
 
By summing up all these we obtain Theorem~\ref{the:main}.

\section{1-Mismatch EDSM}\label{sec:ham}

In this section, we give an alternative to the construction presented in Section~\ref{sec: anchor case}, in the case of \textsc{$1$-Mismatch EDSM}. We do so by finding matches in a tree containing both suffixes of $P$ and elements from the segment $\tilde{T}[i]$, as well as modified versions of those strings. The number of additional strings is bounded by using the \emph{heavy-light decomposition} of Sleator and Tarjan~\cite{sleator_data_1983}. The construction is directly inspired by the one presented by Thankachan et al.~in~\cite{thankachan_provably_2016}, which is itself inspired by the \emph{$k$-errata tree} construction introduced by Cole et al.~in~\cite{cole_dictionary_2004} for indexing with errors. 
We give an algorithm to find all occurrences of $P$ in $\tilde{T}$ with $1$ mismatch by computing sets $\oneAP_i$ under Hamming distance, which, combined with the previously developed techniques, results in solving the \textsc{$1$-Mismatch EDSM} problem in $\cO(nm^2+N\log m)$ time.

Let us start with the following basic definition.

\begin{definition}[\cite{sleator_data_1983}]\label{def:hp}
Let $\mathcal{T}$ be a rooted tree.
The \emph{heavy path} of $\mathcal{T}$ is the path that starts at the root and at
each node descends to the child (called \emph{heavy node}) with the largest number of leaf nodes in its subtree (ties are broken arbitrarily). The \emph{heavy-light decomposition} of $\mathcal{T}$ is defined recursively as a union of the heavy path of $\mathcal{T}$ and the
heavy path decompositions of the off-path subtrees of the heavy path. The nodes that are not heavy nodes are called \emph{light nodes} (the root of $\mathcal{T}$ is always a light node). An edge on a heavy path is called \emph{heavy}; and the other edges are called \emph{light}.
\end{definition}

A crucial property following from Definition~\ref{def:hp} is that any root-to-leaf path crosses $\cO(\log |\mathcal{T}|)$  paths. Each light edge on a path from the root decreases the size of the descending subtree by at least half. Thus the number of light edges on a path from any node to the root is $\cO(\log |\mathcal{T}|)$.

We use the above properties to efficiently construct a tree $\mathcal{T}_{1}(P,\tilde{T}[i])$ (for a given ED text $\tilde{T}[1\dd n]$ of size $N$, a pattern $P[1\dd m]$ and an index $1\le i \le n$ with $||T[i]||=N_i$) in three steps (inspect Figure~\ref{fig:k-errata}):

\begin{description}
    \item[Step 1] We construct the compacted trie containing the strings in $\tilde{T}[i]$ and suffixes $P[j+1\dd m]$ of $P$ for each $j\in AP_{i-1}$. We call this set of suffixes of $P$ $act_{i-1}(P)$. We also add labels $(\iota(X),\#)$ to each node in the tree corresponding to a string $X$ in $act_{i-1}(P)\cup \tilde{T}[i]$, where $\iota(X)$ is a pointer to $X$ and $\#$ is a special label. This takes $\cO(m+N_i)$ time and space~\cite{DBLP:conf/focs/Farach97} (we add the suffixes of $P$ in $\cO(m)$ total time by constructing the suffix tree of $P$ and truncating the superfluous suffixes). We call $\mathcal{T}_0(P,\tilde{T}[i])$ the tree we obtain from this step.
    In the next steps it will be extended with new nodes and labels to obtain $\mathcal{T}_1(P,\tilde{T}[i])$. 
    \item[Step 2] We compute a heavy-light decomposition~\cite{sleator_data_1983} of $\mathcal{T}_{0}(P,\tilde{T}[i])$, which takes time linear in its size, namely $\cO(m+N_i)$.
    \item[Step 3] For each light node $u$ of $\mathcal{T}_{0}(P,\tilde{T}[i])$ let $u'$ be the leaf on the heavy path starting at $u$. 
   Leaf $u'$ corresponds to a string $X$, and for each labeled descendant $v$ of $u$ outside of the heavy path $u \ldots u'$, if $Y$ is the string corresponding to $v$, we compute $p=1+\textsf{LCP}(X,Y)$ (in $\cO(1)$ time after linear-time preprocessing of the tree for LCA queries~\cite{DBLP:conf/latin/BenderF00}) and add to $\mathcal{T}_1(P,\tilde{T}[i])$ the string obtained from $Y$ by replacing $Y[p]$ with $X[p]$, with a label $(\iota(Y),p)$ (a given node can store multiple labels). Intuitively, $p$ is the position of a \emph{mismatch} between (a prefix of) $X$ and (a prefix of) $Y$. Since the tree $\mathcal{T}_0(P,\tilde{T}[i])$ has $\cO(m+N_i)$ nodes and each of them has $\cO(\log(m+N_i))$ light node ancestors, there are no more than $\cO((m+N_i)\log(m+N_i))$ additional nodes and labels. Also the construction of new nodes can be done each time in $\cO(1)$ time, because we in fact just copy a subtree of a light node and merge it with the subtree of its heavy sibling. We have thus arrived at the following lemma.
\end{description}


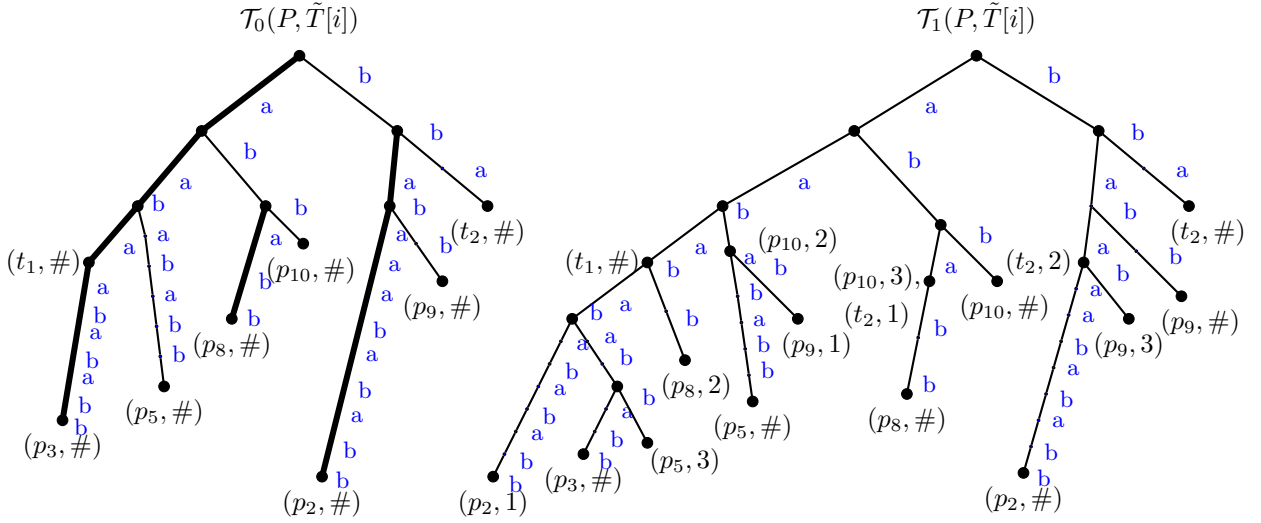
\begin{figure}
    \centering
    \begin{tikzpicture}[yscale=0.5,xscale=0.5,auto,node distance=0.1cm]

\node[] at (0,1) {$\mathcal{T}_{0}(P,\tilde{T}[i])$};
\tikzstyle{dot}=[inner sep=0.0cm, circle, draw]
\node[dot,fill=blue] (t) at (0,0) {};
\foreach \name/\dx/\dy/\parent/\l in {
  l/-2.6/2/t/a,
  ll/-1.7/2/l/a,
  lll/-1.3/1.5/ll/a,
  lll1/-0.1/0.6/lll/a,
  lll2/-0.1/0.6/lll1/b,
  lll3/-0.1/0.6/lll2/a,
  lll4/-0.1/0.6/lll3/b,
  lll5/-0.1/0.6/lll4/a,
  lll6/-0.1/0.6/lll5/b,
  lll7/-0.1/0.6/lll6/b,
  llr/0.2/0.8/ll/b,
  llr2/0.1/0.8/llr/a,
  llr3/0.1/0.8/llr2/b,
  llr4/0.1/0.8/llr3/a,
  llr5/0.1/0.8/llr4/b,
  llr6/0.1/0.8/llr5/b,
  lr/+1.7/2/l/b,
  lrl/-0.3/1/lr/a,
  lrl2/-0.3/1/lrl/b,
  lrl3/-0.3/1/lrl2/b,
  lrr/+1/1/lr/b,
  r/+2.6/2/t/b,
  rl/-0.2/2/r/a,
  rl2/-0.2/0.8/rl/a,
  rl3/-0.2/0.8/rl2/a,
  rl4/-0.2/0.8/rl3/a,
  rl5/-0.2/0.8/rl4/b,
  rl6/-0.2/0.8/rl5/a,
  rl7/-0.2/0.8/rl6/b,
  rl8/-0.2/0.8/rl7/a,
  rl9/-0.2/0.8/rl8/b,
  rl10/-0.2/0.8/rl9/b,
  rlr/+0.7/1/rl/b,
  rlr2/+0.7/1/rlr/b,
  rr/+1.2/1/r/b,
  rr2/+1.2/1/rr/a
} {
  \node[dot,fill=blue] (\name) at ($(\parent)+(\dx,-1*\dy)$) {};
  \draw[black,thick] (\parent)--(\name) node[midway,auto] {\textcolor{blue}{\small \l}};
}

\tikzstyle{dot}=[inner sep=0.05cm, circle, draw]
\foreach \name/\place in {
At/t,Al/l,All/ll,Alll/lll,Alll7/lll7,Allr/llr6,Alr/lr,Alrl/lrl3,Alrr/lrr,Ar/r,Arl/rl10,Arl/rl,Arlr2/rlr2,Arr/rr2}{
\node[dot,fill=black] (\name) at (\place) {};
}

\foreach \na/\nb in {
t/l,l/ll,ll/lll,lll/lll7,lr/lrl3,r/rl,rl/rl10}{
\draw[black,line width=0.8mm] (\na)--(\nb);
}


\foreach \place/\l/\xshift/\yshift in {
lll/${(t_1,\#)}$/-1.2/0,lll7/${(p_3,\#)}$/0/-0.7,llr6/${(p_5,\#)}$/0/-0.7,lrl3/${(p_8,\#)}$/0/-0.7,lrr/${(p_{10},\#)}$/0.2/-0.7,rl10/${(p_2,\#)}$/0/-0.7,rlr2/${(p_9,\#)}$/0/-0.7,rr2/${(t_2,\#)}$/0/-0.7}{
\node[] at ($(\place)+(\xshift,\yshift)$) {\l};
}

\node[] at (18,1) {$\mathcal{T}_{1}(P,\tilde{T}[i])$};
\tikzstyle{dot}=[inner sep=0.0cm, circle, draw]
\node[dot,fill=blue] (t) at (18,0) {};
\foreach \name/\dx/\dy/\parent/\l in {
  l/-3.25/2/t/a,
  ll/-3.5/2/l/a,
  lll/-2/1.5/ll/a,
  llll/-2/1.5/lll/a,
  llll1/-0.3/0.6/llll/a,
  lllll/-0.3/0.6/llll1/b,
  lllll1/-0.3/0.6/lllll/a,
  lllll2/-0.3/0.6/lllll1/b,
  lllll3/-0.3/0.6/lllll2/a,
  lllll4/-0.3/0.6/lllll3/b,
  lllll5/-0.3/0.6/lllll4/b,
  llllr/0.4/0.6/llll/b,
  llllr1/0.4/0.6/llllr/a,
  llllr2/0.4/0.6/llllr1/b,
  llllrl/-0.3/0.6/llllr2/a,
  llllrl1/-0.3/0.6/llllrl/b,
  llllrl2/-0.3/0.6/llllrl1/b,
  llllrr/0.8/1.5/llllr2/b,
  lllr/0.5/1.3/lll/b,
  lllr1/0.5/1.3/lllr/b,  
  llr/0.2/1.2/ll/b,
  llr2/0.2/1.2/llr/a,
  llr3/0.1/0.7/llr2/b,
  llr4/0.1/0.7/llr3/a,
  llr5/0.1/0.7/llr4/b,
  llr6/0.1/0.7/llr5/b,
  llrr/1.8/1.8/llr/b,
  lr/+2.3/2.5/l/b,
  lrl/-0.3/1.5/lr/a,
  lrl2/-0.3/1.5/lrl/b,
  lrl3/-0.3/1.5/lrl2/b,
  lrr/+1.5/1.5/lr/b,
  r/+3.25/2/t/b,
  rl/-0.2/2/r/a,
  rl2/-0.2/1.5/rl/a,
  rl3/-0.2/0.7/rl2/a,
  rl4/-0.2/0.7/rl3/a,
  rl5/-0.2/0.7/rl4/b,
  rl6/-0.2/0.7/rl5/a,
  rl7/-0.2/0.7/rl6/b,
  rl8/-0.2/0.7/rl7/a,
  rl9/-0.2/0.7/rl8/b,
  rl10/-0.2/0.7/rl9/b,
  rllr/+1.2/1.5/rl2/b,
  rlr/+1.2/1.2/rl/b,
  rlr2/+1.2/1.2/rlr/b,
  rr/+1.2/1/r/b,
  rr2/+1.2/1/rr/a
} {
  \node[dot,fill=blue] (\name) at ($(\parent)+(\dx,-1*\dy)$) {};
  \draw[black,thick] (\parent)--(\name) node[midway,auto] {\textcolor{blue}{\small \l}};
}

\tikzstyle{dot}=[inner sep=0.05cm, circle, draw]
\foreach \name/\place in {
At/t,Al/l,All/ll,Alll/lll,Allll/llll,Alllll5/lllll5,Allllr2/llllr2,Allllrl2/llllrl2,Allllrr/llllrr,Allr/llr,Allr6/llr6,Allrr/llrr,Alllr1/lllr1,Alr/lr,Alrl/lrl,Alrl3/lrl3,Alrr/lrr,Ar/r,Arl2/rl2,Arl10/rl10,Arllr/rllr,Arlr2/rlr2,Arr2/rr2}{
\node[dot,fill=black] (\name) at (\place) {};
}


\foreach \place/\l/\xshift/\yshift in {
lll/${(t_1,\#)}$/-1.2/0,llr/${(p_{10},2)}$/+1.8/+0.3,llllrl2/${(p_3,\#)}$/0/-0.7,llllrr/${(p_5,3)}$/1/-0.5,lllr1/${(p_8,2)}$/+0.3/-0.7,llr6/${(p_5,\#)}$/0/-0.7,llrr/${(p_9,1)}$/0.5/-0.7,lrl/${(p_{10},3),}$/-1.4/+0.1,lrl/${(t_2,1)}$/-1.4/-0.9,lrl3/${(p_8,\#)}$/0/-0.7,lrr/${(p_{10},\#)}$/0.2/-0.7,lllll5/${(p_2,1)}$/0/-0.7,rl2/${(t_2,2)}$/-1.2/0,rl10/${(p_2,\#)}$/0/-0.7,rllr/${(p_9,3)}$/0/-0.7,rlr2/${(p_9,\#)}$/0.5/-0.7,rr2/${(t_2,\#)}$/0.5/-0.7}{
\node[] at ($(\place)+(\xshift,\yshift)$) {\l};
}

\end{tikzpicture}
\caption{$\mathcal{T}_0(P,\tilde{T}[i])$ and $\mathcal{T}_1(P,\tilde{T}[i])$ for the example from Figure~\ref{fig:rectangles} ($P=bbaaaabababb$, $AP_{i-1}=\{1,2,4,7,8,9\}$, $\tilde{T}[i]=\{aaa,bba\}$) with labels ($p_j=\iota(P[j\dd m]), t_j=\iota(\tilde{T}[i][j])$) and heavy paths.}\label{fig:k-errata}
\end{figure}

\begin{lemma}\label{Lem:tree size}
The construction of $\mathcal{T}_1(P,\tilde{T}[i])$ takes $\cO((m+N_i)\log(m+N_i))$ time and space.
\end{lemma}

We now prove that the tree $\mathcal{T}_{1}(P,\tilde{T}[i])$ satisfies the following property.

\begin{lemma}\label{Lem:string nodes condition}

Let $X\in act_{i-1}(P)$.
A string $Y\in \tilde{T}[i]$ is at Hamming distance at most $1$ from a prefix of $X$ having length $|Y|$ if and only if $\mathcal{T}_1(P,\tilde{T}[i])$ contains two nodes $u$, $v$ respectively labeled by $(\iota(X),p)$ and $(\iota(Y),p')$, for some $p,p'\in \mathbb{N}\cup\{\#\}$, such that $u$ is a descendant of $v$, and one of the following is satisfied:
\begin{itemize}
    \item $p=p'\in\mathbb{N}$
    \item $p=\#$ or $p'=\#$.
\end{itemize}

\end{lemma}

\begin{proof}
For the forward implication, if $Y$ is a prefix of $X$ then the claim is trivial since $\mathcal{T}_0(P,\tilde{T}[i])$ contains nodes with labels $(\iota(X),\#)$ and $(\iota(Y),\#)$, and thus the first node is a descendant of the second one.
Now, we assume that $Y$ has one mismatch with $X'=X[1\dd |Y|]$ at a position $p$.
Let $u$, $v$ be nodes in $\mathcal{T}_{0}(P,\tilde{T}[i])$ respectively corresponding to $X$ and $Y$, and let $w$ be their lowest common light ancestor. Let $Z$ be the string corresponding to the leaf on the heavy path starting at the heavy child of $w$. 
Since $X$ and $Y$ have a mismatch at position $p$, at least one of them has a mismatch with $Z$ at position $p$ and there are no mismatches to the left of $p$.
Indeed, suppose towards a contradiction that there exists some $p'<p$ such that $Z[p']\neq X[p'](=Y[p'])$: then the node corresponding to $X[1\dd p'](=Y[1\dd p'])$ would not be on the heavy path corresponding to $Z$, but would be a common ancestor of $u$ and $v$, and thus $w$ would not be the lowest common light ancestor of $u$ and $v$, a contradiction.

Assume first that $X[p]\neq Z[p]$. Then, there is a node with a label $(\iota(X),p)$ in the tree, which is a descendant of either $v$, having label $(\iota(Y),\#)$ (if $Y[p]=Z[p]$), or a node having label $(\iota(Y),p)$ (if $Y[p]\neq Z[p]$), because we assumed that $X$ and $Y$ do not have any other mismatch. Finally, if $X[p]=Z[p]$, then $Y[p]\neq Z[p]$ and the node with label $(\iota(Y),p)$ is an ancestor of $u$, having label $(\iota(X),\#)$. 

To prove the reverse implication, let us assume that the consequences are satisfied. Let $u$ be the node whose label contains $(\iota(X),p)$ and $v$ the node whose label contains $(\iota(Y),p')$. We first assume $p=p'\in\mathbb{N}$.
Note that, by the construction of $\mathcal{T}_1(P,\tilde{T}[i])$, the node $u$ (resp. $v$) corresponds to a string obtained by one letter modification on $X$ (resp. on $Y$) at the same position $p$. We denote the resulting string $\hat{X}$ (resp $\hat{Y}$). Since $v$ is an ancestor of $u$ in $\mathcal{T}_1(P,\tilde{T}[i])$, $\hat{Y}$ is a prefix of $\hat{X}$. But this exactly means that $Y$ has Hamming distance $1$ with the length $|Y|$ prefix of $X$ (or Hamming distance $0$ if both replacements replaced the same letter).
If the second condition is satisfied, namely if $p=\#$ or $p'=\#$, then it means that one replacement in $Y$ gives $\hat{Y}$ which is a prefix of $X$, or that $Y$ is a prefix of $\hat{X}$, which is one replacement away from $X$, therefore we have the claimed result.
\end{proof}

We next formalize how to find nodes satisfying one of the conditions from Lemma~\ref{Lem:string nodes condition} and deduce the approximate active prefixes corresponding to the Anchor Case for segment $\tilde{T}[i]$. Let $v_1\ \mathtt{OR}\ v_2$ denote a bitwise OR of two vectors, and $v_1 \oplus x$ denote vector $v_1$ shifted by $x$ positions to the right (the first $x$ positions are set to $0$).

\begin{algorithm}[ht]\caption{Search($\mathcal{T}$)}\label{Alg: k errata}
\begin{algorithmic}[1]
\State{\textbf{Global variables}: the set $act_{i-1}(P)$, a segment $\tilde{T}[i]$, and bit vectors $V_{\#}$, $V_1$, $\ldots$, $V_m$, $V_{ANY}$, $V_{res}$ all of size $|P|+1$, and initially set to all $0$'s}
\State{\textbf{Input}: $\mathcal{T}$ - a subtree of $\mathcal{T}_1(P,\tilde{T}[i])$ with root $r$}
\State{\textbf{Output}: represented by global bit vector $V_{res}$}
\For{each label $(\iota(X),p)$ with $X\in \tilde{T}[i]$ on $r$}
\State{set $V_p[|X|]$ and $V_{ANY}[|X|]$ to $1$}
\EndFor
\For{each label $(\iota(X),p)$ with $X\in act_{i-1}(P)$ on $r$}
\If{$p=\#$}
\State{update $V_{res}$ to $V_{res}\ \mathtt{OR}\ (V_{ANY}\oplus (m-|X|))$.}
\Else{~update $V_{res}$ to $V_{res}\ \mathtt{OR}\ ((V_p\ \mathtt{OR}\ V_{\#} ) \oplus (m-|X|))$}
\EndIf
\EndFor
\For{each $\mathcal{T}'$ a subtree of $\mathcal{T}$ rooted at a child of $r$}
\State{run Search($\mathcal{T}'$)}
\EndFor
\For{each label $(\iota(X),p)$ with $X\in \tilde{T}[i]$ on $r$}
\State{set $V_p[|X|]$ and $V_{ANY}[|X|]$ to $0$}
\EndFor
\end{algorithmic}
\end{algorithm}

\begin{proposition}\label{prop:k-errata}
Algorithm~\ref{Alg: k errata} with input $\mathcal{T}_1(P,\tilde{T}[i])$ returns $V_{res}$ such that $V_{res}[p]=1$ if and only if $p$ is an element of $\oneAP_i$ corresponding to the Anchor Case for segment $\tilde{T}[i]$. Algorithm~\ref{Alg: k errata} runs in $\cO((m+N_i)\log(m+N_i)+m^2)$ time.
\end{proposition}
\begin{proof}
We first need the following remark: if $P[1\dd k]$ extends into $P[1\dd k']$ in $\tilde{T}[i]$, that means that some $Y\in \tilde{T}[i]$ is at Hamming distance 1 from the prefix $P[k+1\dd k']$ of $P[k+1\dd |P|]$. Therefore, we are looking for the pairs described in Lemma~\ref{Lem:string nodes condition}. We show then that $V_{res}[k']=1$ after the end of the procedure if and only if there is a pair of nodes $u,v$ in $\mathcal{T}_{1}(P,\tilde{T}[i])$ satisfying the conditions of Lemma~\ref{Lem:string nodes condition} for $X=P[k+1\dd |P|]$, $Y\in \tilde{T}[i]$, and $|Y|=k'-k$.

Let us assume the existence of such a pair $(u,v)$. Since the tree is traversed in a DFS, the node $v$ (with a label $(\iota(Y),p')$, $p'\in\mathbb{N}\cup\{\#\}$) is traversed before $u$, which is its descendant; and at this moment, $V_{p'}[|Y|]$ is set to $1$, as well as $V_{ANY}[|Y|]$. Since $u$ is a descendant of $v$, the vectors are not modified at position $|Y|$ until $u$ is visited: that would mean that $v$ has a strict descendant representing a string of the same length as the string represented by $v$. When the label $(\iota(X),p)$ for $X\in act_{i-1}(P)$ is visited on $u$, we set the position $(m-|X|)+|Y|=k'$ of $V_{res}$ to $1$ if $V_p[|Y|]=1$ or $V_{\#}[|Y|]=1$ (which happens if $p=p'\in\mathbb{N}\cup\{\#\}$ or if $p'=\#$) and when $p=\#$ if $V_{ANY}[|Y|]=1$.

Vice versa, if after the processing one has $V_{res}[k']=1$, this means that at some point in the DFS a node $u$ having a label $(\iota(X),p)$  with $X\in act_{i-1}(P)$ and $p\in\mathbb{N}\cup\#$ was visited, and that at this point, for $k=m-|X|$, one had $V_{ANY}[k'-k]=1$ or $V_{p'}[k'-k]=1$ for $(p,p')$ satisfying the conditions from Lemma~\ref{Lem:string nodes condition}. This one had to be set previously in the DFS at a node $v$ having label $(\iota(Y),p')$ for $Y\in \tilde{T}[i]$ with $|Y|=k'-k$. Finally,  an ancestor of $u$ can be chosen as such $v$, because otherwise, from the DFS traversal order, the corresponding component of the vectors would have been set to $0$. Now, the pair of nodes $(u,v)$ satisfy the conditions of Lemma~\ref{Lem:string nodes condition}, and from our observations that means that there is an active prefix with $1$ error of $P$ having length $k$, extending up to $\tilde{T}[i]$.

The running time follows from the fact that the algorithm visits only $\cO((m+N_i)\log(m+N_i))$ labels by Lemma~\ref{Lem:tree size}, and from the fact that the tree is traversed in a DFS. The analysis of each label consists in reading it and doing a constant number of bit modifications in the stored vectors, and, for $\cO(m \log(m+N_i))$ of them (the one corresponding to a suffix of $P$), doing an $\mathtt{OR}$ operation which takes $\cO(\frac{m}{\log (N+m)})$ time in the word RAM model. This gives us the required running time.
\end{proof}

\begin{corollary}
\textsc{$1$-Mismatch EDSM} can be solved in $\cO(nm^2+N\log m)$ time.
\end{corollary}
\begin{proof}
We proceed in the same way as in the reporting version of Section~\ref{sec:wrapping up}; the only difference is that, when $N_i\le m^3$, to extend $AP_{i-1}$ into $\oneAP_i$, instead of using the $\cO((m^2+N_i)\log m)$-time algorithm, we use the one from Proposition~\ref{prop:k-errata}. Due to this change, the algorithm runs in the desired time.
Indeed, notice that when $N_i\le m^3$, $\cO((N_i+m)\log(m+N_i)+m^2)=\cO(N_i\log m+m^2)$, and when $N_i\ge m^3$, $\cO(m^3+N_i)=\cO(N_i)$. The total time is thus $\cO(nm^2+N\log m)$.
\end{proof}

\section{Open Questions}\label{sec:con}

While our techniques (Sections~\ref{sec:edit} and~\ref{sec:ham}) seem to generalize relatively easily to $k$ errors,
they would incur some exponential factor with respect to $k$.
We leave the following basic questions open:

\begin{enumerate}
    \item Can we design an $\cO(nm^2 + N)$-time algorithm for $1$-EDSM under edit or Hamming distance?
    \item Can our techniques be efficiently generalized for $k>1$ errors or mismatches? 
    \item Can our Hamming distance improvement for 1 mismatch (Section~\ref{sec:ham}) be extended to edit distance?
\end{enumerate}

\bibliographystyle{plain}
\bibliography{references}
\end{document}